\documentclass[aps,pra,10pt,twocolumn,superscriptaddress,showkeys,letterpaper]{revtex4-2}

\pdfoutput=1

\usepackage[T1]{fontenc}
\usepackage[utf8]{inputenc}
\usepackage{color}
\usepackage{babel}
\usepackage{mathtools}
\usepackage{amsmath}
\usepackage{amsthm}
\usepackage{amssymb}
\usepackage[pdfusetitle,
 bookmarks=true,bookmarksnumbered=false,bookmarksopen=false,
 breaklinks=false,pdfborder={0 0 0},pdfborderstyle={},backref=false,colorlinks=true]
 {hyperref}
\hypersetup{
 linkcolor=magenta, urlcolor=blue, citecolor=blue}

\makeatletter
\theoremstyle{plain}
\newtheorem{thm}{\protect\theoremname}
\theoremstyle{plain}
\newtheorem{cor}[thm]{\protect\corollaryname}
\theoremstyle{plain}
\newtheorem{lem}[thm]{\protect\lemmaname}
\theoremstyle{plain}
\newtheorem{prop}[thm]{\protect\propositionname}
\theoremstyle{remark}
\newtheorem{rem}[thm]{\protect\remarkname}

\DeclareMathOperator{\Tr}{Tr}

\DeclareMathOperator{\Herm}{Herm}
\DeclareMathOperator{\id}{id}

\DeclareMathOperator{\softmin}{softmin}
\DeclareMathOperator{\signum}{sgn}

\allowdisplaybreaks

\makeatother

\providecommand{\corollaryname}{Corollary}
\providecommand{\lemmaname}{Lemma}
\providecommand{\propositionname}{Proposition}
\providecommand{\remarkname}{Remark}
\providecommand{\theoremname}{Theorem}

\begin{document}
\title{Unifying quantum measurement constructions via\\a relative-entropy
minimum change principle}

\author{Nana Liu}
\affiliation{Institute of Natural Sciences, School of Mathematical Sciences, Ministry of Education Key Laboratory in Scientific and Engineering Computing, and Global College, Shanghai Jiao Tong University, Shanghai 200240, China}

\author{Mark M. Wilde}
\affiliation{School of Electrical and Computer Engineering,
Cornell University, Ithaca, New York 14850, USA}

\begin{abstract}
The minimum change principle provides an information-theoretic characterization
of the Bayes reversal channel in classical probability theory and
has recently been proposed as a framework for extending Bayes' rule
to quantum information theory. Using quantum relative entropy, we
investigate a minimum change principle for the setting of quantum statistical
inference. Specifically, we consider a forward process based on a
classical-to-quantum preparation channel and a reverse process based
on a quantum-to-classical measurement channel. We establish a closed-form
characterization of measurements that are optimal for this principle, and this optimal measurement can be found via a dual formulation involving a \textit{single} unconstrained Hermitian
variable. This perspective allows us to recover some notable measurements within the same framework, including
pretty good measurements and Fermi--Dirac thermal measurements, and we use it to discover a novel family that we call softmin thermal measurements. We further
show that softmin thermal measurements arise as optimal solutions
to entropy-regularized semidefinite optimization problems, demonstrating
that they play a role for measurements analogous to that of thermal
states in statistical mechanics. Finally, we prove an additivity property
for the relative-entropy minimum change principle and investigate
the performance of Fermi--Dirac thermal measurements for quantum
hypothesis testing.
\end{abstract}

\date{\today}
\maketitle

\tableofcontents{}

\section{Introduction}

\subsection{Background}

Bayes' rule is one of the foundational principles of probability theory
and has become indispensable throughout statistics, physics, and machine
learning (see, e.g.,~\cite{Jaynes2003}). It relates a forward probabilistic
model to its reverse by expressing the same joint probability distribution
in two equivalent ways: one in terms of a prior distribution and a
forward channel, and the other in terms of posterior probabilities
and a reverse channel. 

Extending Bayes' rule to quantum information theory has been a longstanding
goal~\cite{SchackBrunCaves2001,LeiferSpekkens2013,ParzygnatRusso2022,ParzygnatFullwood2023,CenxinEtAl2023},
and several quantum generalizations of it have been proposed. Among
the most influential are the pretty good measurement~\cite{Belavkin75,Belavkin75a,hughston1993complete,hausladen1994pretty}
and its generalization, the Petz recovery channel~\cite{Petz1986,Petz1988},
both of which admit compelling Bayesian interpretations and have found
numerous applications in quantum state discrimination and quantum
error correction~\cite{BarnumKnill2002,HaydenJozsaPetzWinter2004,JungeRennerSutterWildeWinter2018},
among others.

Quite recently, the authors of~\cite{Bai2025} put forward a framework
for extending Bayes' rule to quantum information theory based on a
quantum generalization of the minimum change principle~\cite{Aitchison1975,MayHarper1976,Williams1980,Zellner1988}.
Their approach considers arbitrary quantum channels and reverse channels,
and it identifies an optimal reverse process by minimizing a divergence
between bipartite states describing the forward and reverse dynamics.
In particular, see \cite[Eqs.~(1)--(5)]{Bai2025} for a concise description
of the approach in the classical case and \cite[Eqs.~(6)--(15)]{Bai2025}
for their quantum generalization. Using the quantum fidelity~\cite{Uhlmann1976}
as the divergence, they proved that the Petz recovery channel and
generalizations of it are optimal reverse processes while providing
a unified conceptual framework.

\subsection{Summary of contributions}

Motivated by the framework of~\cite{Bai2025}, we investigate a
minimum change principle in the particular setting of quantum statistical
inference~\cite{Helstrom1976,Holevo2011,Hayashi2005Asymptotic}. Specifically,
we consider forward processes described by classical-to-quantum preparation
channels and reverse processes described by quantum-to-classical measurement
channels. We further adopt quantum relative entropy~\cite{Umegaki1962},
a key information-theoretic distinguishability measure, as the measure
of change, rather than fidelity. These two specializations lead to
a convex optimization problem and yield a novel family of optimal
measurements. 

As one of our key findings (Theorem~\ref{thm:main}), we show that
the optimization admits a simple dual formulation involving only a
single Hermitian variable. The dual objective is strictly convex and
possesses a unique minimizer $A^{\star}$, which satisfies the following
nonlinear operator equation:
\begin{equation}
\sum_{x\in\mathcal{X}}e^{\ln\sigma_{x}-A^{\star}}=\tau,\label{eq:nonlinear-op-eq}
\end{equation}
where $\mathcal{X}$ is an alphabet, $\sigma_{x}=p(x)\rho_{x}$, the
probability distribution $p(x)$ represents the prior, $\rho_{x}$
is the state resulting from inputting the symbol $x$ to the forward
classical-to-quantum channel, and following~\cite{Bai2025}, we call
$\tau$ the reference state for the reverse process. From the solution
of~\eqref{eq:nonlinear-op-eq}, we obtain a closed-form expression
for the optimal measurement $\left(M_{x}^{\star}\right)_{x\in\mathcal{X}}$,
given by
\begin{equation}
M_{x}^{\star}=\tau^{-\frac{1}{2}}\left(e^{\ln\sigma_{x}-A^{\star}}\right)\tau^{-\frac{1}{2}}.\label{eq:optimal-meas}
\end{equation}
By applying~\eqref{eq:nonlinear-op-eq}, we see that the optimal measurement
satisfies the completeness relation $\sum_{x\in\mathcal{X}}M_{x}^{\star}=I$,
so that it is indeed a legitimate measurement. To the best of our
knowledge, this family of measurements has not previously appeared
in the quantum information literature.

An appealing consequence of the characterization in~\eqref{eq:optimal-meas}
is that the reference state $\tau$ parametrizes a continuous family
of optimal measurements. Different choices of $\tau$ recover the
pretty good measurement~\cite{Belavkin75,Belavkin75a,hughston1993complete,hausladen1994pretty}
and the Fermi--Dirac thermal measurement~\cite{liu2026}, the latter
introduced recently in semidefinite optimization and used in quantum
machine learning~\cite{he2026fermidiracmachines,he2026canonical}
(see also~\cite{Lindsey2023}). We also introduce a novel measurement
that we call the softmin thermal measurement, which generalizes the
classical softmin decision rule~\cite{Bridle1990,Bishop2006,Goodfellow2016}.
Thus, the minimum change principle developed here provides a common
variational principle underlying these seemingly distinct measurement
constructions.

Given the critical role of thermal states in physics, quantum optimization,
and quantum machine learning, we also consider the case when the forward
process consists of preparing thermal states and the reference state
$\tau$ is maximally mixed. Doing so clarifies how the softmin thermal
measurement represents a multiclass quantum generalization of the
classical softmin rule and how the Fermi--Dirac thermal measurement
is a quantum generalization of the sigmoid activation function.

The appearance of the softmin thermal measurement naturally raises
the question of whether it is specific to the minimum change principle
or whether it also arises independently from other optimization problems.
We answer this question affirmatively by proving that these measurements
are optimal solutions to a broad class of entropy-regularized semidefinite
optimization problems.

Next, we explore connections of our findings in quantum information
theory. To this end, we prove an additivity property for the minimum
change principle in Theorem~\ref{thm:main}, noting that it is conceptually
similar to the additivity of accessible information~\cite{Holevo1973StatisticalDecision}.
As an implication, for a forward process consisting of a product of
two processes and a reference state that is a product of two reference
states, the optimal reversal channel is a product of the reversal
channels that are optimal for each individual case.

Finally, we characterize the performance of Fermi--Dirac thermal
measurements in quantum hypothesis testing, both in the non-asymptotic
and asymptotic scenarios. In particular, we establish an upper bound
on the error probability of hypothesis testing when using Fermi--Dirac
thermal measurements, and we extend this result to prove a lower bound
on the error exponent of this task in the asymptotic scenario. We
also prove an upper bound on the error exponent, which indicates that
this measurement strategy cannot generally achieve optimal performance.
As a corollary, we conclude that Fermi--Dirac thermal measurements
cannot generally be ``pretty good,'' in the sense of~\cite{BarnumKnill2002};
that is, there cannot exist a universal multiplicative constant relating
the optimal hypothesis testing error probability to the error probability
when using Fermi--Dirac thermal measurements.

\subsection{Paper organization}

The rest of our paper is organized as follows. In Section~\ref{sec:Minimum-change-principles},
we develop relative-entropy based minimum change principles and prove
our main theorem for one of these principles (Theorem~\ref{thm:main}).
Therein, we also outline an algorithm for computing the optimal Hermitian
operator for the dual formulation of this minimum change principle
(Section~\ref{subsec:Computing-the-optimal-dual}). In Section~\ref{sec:special-cases},
we consider various special cases of Theorem~\ref{thm:main}, which
include the classical case, the pretty good measurement, the softmin
thermal measurement, and the Fermi--Dirac thermal measurement of
\cite{liu2026}. In Section~\ref{sec:Bayesian-minimal-change-thermal-states},
we specialize Theorem~\ref{thm:main} to thermal states and argue
therein especially how the softmin thermal measurement and the Fermi--Dirac
thermal measurement represent quantum generalizations of the softmin
and sigmoid functions, respectively, the latter having played prominent
roles in classical machine learning. In Section~\ref{sec:Softmin-thermal-measurements-SDPs},
we show how softmin thermal measurements arise as optimal solutions
to entropically-regularized semidefinite optimization problems, representing
a broad generalization of the results of~\cite{liu2026}. In Section
\ref{sec:Connections-to-QIT}, we connect our findings to quantum
information theory, proving an additivity result for the minimum change
principle in Theorem~\ref{thm:main} and discussing the performance
of Fermi--Dirac thermal measurements in non-asymptotic and asymptotic
quantum hypothesis testing. Finally, in Section~\ref{sec:Conclusion},
we conclude with a brief summary of our findings and suggestions for
future directions. 

\section{Minimum change principles for quantum measurements}

\label{sec:Minimum-change-principles}

We begin by developing a general formulation of relative-entropy minimum
change principles and follow in Section~\ref{sec:special-cases} by
specializing one of them to several cases of interest.

\subsection{Forward process}

Consider a forward process, which receives a classical input symbol
$x\in\mathcal{X}$ and outputs a $d$-dimensional quantum state $\rho_{x}$,
where $d\in\mathbb{N}$. Set
\begin{equation}
r\equiv|\mathcal{X}|.
\end{equation}
The forward process is known as a classical-to-quantum channel $x\to\rho_{x}$
that outputs the state $\rho_{x}$ upon receiving the classical input
$x$. Equivalently, this process is described by the following quantum
channel:
\begin{equation}
\mathcal{E}(\omega)\coloneqq\sum_{x\in\mathcal{X}}\langle x|\omega|x\rangle\rho_{x},\label{eq:forward-channel}
\end{equation}
where $\omega$ is an arbitrary input state and $\left\{ |x\rangle\right\} _{x\in\mathcal{X}}$
is an orthonormal basis corresponding to the alphabet $\mathcal{X}$.
The Choi operator of this channel is given by
\begin{align}
\Gamma^{\mathcal{E}} & \coloneqq(\mathcal{E}\otimes\id)(\Gamma^{(r)})\\
 & =\sum_{x\in\mathcal{X}}\rho_{x}\otimes|x\rangle\!\langle x|,
\end{align}
where the maximally entangled operator $\Gamma^{(r)}$ is defined
as
\begin{equation}
\Gamma^{(r)}\coloneqq\sum_{x,x'\in\mathcal{X}}|x\rangle\!\langle x'|\otimes|x\rangle\!\langle x'|.\label{eq:max-ent-vec}
\end{equation}
A prior probability distribution $\left(p(x)\right)_{x\in\mathcal{X}}$
over the input symbols in $\mathcal{X}$ is described by the following
state:
\begin{equation}
\pi_{r}\coloneqq\sum_{x\in\mathcal{X}}p(x)|x\rangle\!\langle x|.
\end{equation}

The bipartite state $Q_{\text{fwd}}$ defined in \cite[Eq.~(8)]{Bai2025},
which represents the forward process, is defined as
\begin{align}
Q_{\text{fwd}} & \coloneqq\left(I\otimes\pi^{\frac{1}{2}}\right)\Gamma^{\mathcal{E}}\left(I\otimes\pi^{\frac{1}{2}}\right)\\
 & =\sum_{x\in\mathcal{X}}\sigma_{x}\otimes|x\rangle\!\langle x|,\label{eq:q-fwd-gen-1}
\end{align}
where the substate $\sigma_{x}$ is given by 
\begin{equation}
\sigma_{x}\coloneqq p(x)\rho_{x}.\label{eq:sigma-sub-state-def}
\end{equation}
We also define the average output state as follows:
\begin{equation}
\sigma\coloneqq\sum_{x\in\mathcal{X}}\sigma_{x}=\sum_{x\in\mathcal{X}}p(x)\rho_{x}.\label{eq:average-sigma-state}
\end{equation}

\subsection{Reverse process}

A reverse process is represented by a quantum--to--classical channel
(i.e., measurement channel) of the following form:
\begin{equation}
\mathcal{M}(\gamma)\coloneqq\sum_{x\in\mathcal{X}}\Tr\!\left[M_{x}\gamma\right]|x\rangle\!\langle x|,\label{eq:meas-channel-gen-1}
\end{equation}
where $\gamma$ is an input state and $\left(M_{x}\right)_{x\in\mathcal{X}}$
is a positive operator-valued measure (POVM). That is, $\left(M_{x}\right)_{x\in\mathcal{X}}$
satisfies $M_{x}\geq0$ for all $x\in\mathcal{X}$ and $\sum_{x\in\mathcal{X}}M_{x}=I$.
Define the following shorthand for the POVM $\left(M_{x}\right)_{x\in\mathcal{X}}$:
\begin{equation}
M_{\mathcal{X}}\equiv\left(M_{x}\right)_{x\in\mathcal{X}}.\label{eq:meas-shorthand}
\end{equation}

Recall from \cite[Eq.~(C13)]{liu2026}, e.g., that the Choi operator
$\Gamma^{\mathcal{M}}$ corresponding to the measurement channel $\mathcal{M}$
is given by
\begin{equation}
\Gamma^{\mathcal{M}}=\sum_{x\in\mathcal{X}}M_{x}^{T}\otimes|x\rangle\!\langle x|.
\end{equation}
As discussed in \cite[Eq.~(12)]{Bai2025}, we can choose a reference
state $\tau$ and define the bipartite state $Q_{\mathrm{rev}}(M_{\mathcal{X}})$,
representing the reverse process, as follows:
\begin{align}
& Q_{\mathrm{rev}}(M_{\mathcal{X}}) \notag\\
& \coloneqq\left[\left(\left(\tau^{\frac{1}{2}}\right)^{T}\otimes I\right)\Gamma^{\mathcal{M}}\left(\left(\tau^{\frac{1}{2}}\right)^{T}\otimes I\right)\right]^{T}\\
 & =\left(\sum_{x\in\mathcal{X}}\left(\tau^{\frac{1}{2}}\right)^{T}M_{x}^{T}\left(\tau^{\frac{1}{2}}\right)^{T}\otimes|x\rangle\!\langle x|\right)^{T}\\
 & =\sum_{x\in\mathcal{X}}\tau^{\frac{1}{2}}M_{x}\tau^{\frac{1}{2}}\otimes|x\rangle\!\langle x|.\label{eq:q-rev-simple}
\end{align}

\subsection{Minimum change principles based on quantum relative entropy}

Following \cite[Eq.~(15)]{Bai2025}, we can formulate at least two
minimum change principles based on quantum relative entropy:
\begin{align}
\min_{M_{\mathcal{X}}}D(Q_{\mathrm{rev}}(M_{\mathcal{X}})\|Q_{\text{fwd}}),\label{eq:minimal-change-I}\\
\min_{M_{\mathcal{X}}}D(Q_{\text{fwd}}\|Q_{\mathrm{rev}}(M_{\mathcal{X}})),\label{eq:minimal-change-II}
\end{align}
where $Q_{\mathrm{rev}}(M_{\mathcal{X}})$ is given by~\eqref{eq:q-rev-simple}
and $Q_{\text{fwd}}$ by~\eqref{eq:q-fwd-gen-1}. The minimizations
above are over every POVM $M_{\mathcal{X}}$, and $D(Q_{\mathrm{rev}}(M_{\mathcal{X}})\|Q_{\text{fwd}})$
and $D(Q_{\text{fwd}}\|Q_{\mathrm{rev}}(M_{\mathcal{X}}))$ are based
on the standard (Umegaki) relative entropy~\cite{Umegaki1962}, defined
for states $\xi$ and $\omega$ as
\begin{equation}
D(\xi\|\omega)\coloneqq\Tr\!\left[\xi(\ln\xi-\ln\omega)\right].
\end{equation}

Among the two possible relative-entropy minimum change principles
in~\eqref{eq:minimal-change-I} and~\eqref{eq:minimal-change-II},
we focus nearly exclusively on the minimum change principle in~\eqref{eq:minimal-change-I}
throughout our paper. There is mathematical simplicity in doing so,
along with historical precedent, going back to the work of Jaynes
linking information theory and statistical mechanics~\cite{Jaynes1957,Jaynes1957a,Jaynes1962},
in information geometry~\cite{Csiszar1975IDivergence,Csiszar1984Sanov,CsiszarMatus2003InformationProjections},
where it is known as the information projection (or I-projection),
and more recently in quantum information science in several different
contexts~\cite{Nuradha2025MultivariateFidelities,girardi2025quantumumlautinformation,qiu2025quantumesschertransform,Lami2026Asymptotic}.
Here, we can think of the optimal $M_{\mathcal{X}}$ in~\eqref{eq:minimal-change-I}
as the information projection of the forward process onto the convex
set of reverse processes realizable by quantum measurements. A notable
consequence is that the optimizer for~\eqref{eq:minimal-change-I}
is unique and admits an explicit representation, leading to the simple
dual formulation developed in Theorem~\ref{thm:main} below. By contrast,
the opposite ordering $D(Q_{\text{fwd}}\|Q_{\mathrm{rev}}(M_{\mathcal{X}}))$
leads to a substantially different optimization problem and does not
appear to admit a comparable characterization. We explore the associated
difficulties briefly in Appendix~\ref{app:Alternative-minimum-change}.

Our main result is the following characterization of the minimum change
principle in~\eqref{eq:minimal-change-I}:
\begin{thm}
\label{thm:main}Let $\tau$ be a positive definite reference state,
and let $\sigma_{x}>0$ be defined as in~\eqref{eq:sigma-sub-state-def}
for all $x\in\mathcal{X}$. Then the following equalities hold:
\begin{align}
 & \min_{M_{\mathcal{X}}}D(Q_{\mathrm{rev}}(M_{\mathcal{X}})\|Q_{\mathrm{fwd}})\nonumber \\
 & =\min_{M_{\mathcal{X}}}\sum_{x\in\mathcal{X}}D(\tau^{\frac{1}{2}}M_{x}\tau^{\frac{1}{2}}\|\sigma_{x})\label{eq:rewrite-primal}\\
 & =1-\inf_{A\in\Herm}\left\{ \Tr\!\left[A\tau\right]+\sum_{x\in\mathcal{X}}\Tr\!\left[e^{\ln\sigma_{x}-A}\right]\right\} ,\label{eq:dual-bayesian-min-change}
\end{align}
where $M_{\mathcal{X}}$ is defined in~\eqref{eq:meas-shorthand},
$Q_{\mathrm{rev}}(M_{\mathcal{X}})$ is given by~\eqref{eq:q-rev-simple},
and $Q_{\mathrm{fwd}}$ is given by~\eqref{eq:q-fwd-gen-1}. Furthermore,
the dual optimization in~\eqref{eq:dual-bayesian-min-change} admits
a unique minimum $A^{\star}$ that satisfies
\begin{equation}
\sum_{x\in\mathcal{X}}e^{\ln\sigma_{x}-A^{\star}}=\tau,\label{eq:constraint-optimal-povm-1}
\end{equation}
and the unique optimal POVM consists of measurement operators labeled
by $M_{x}^{\star}$ and, for all $x\in\mathcal{X}$, given by
\begin{equation}
M_{x}^{\star}=\tau^{-\frac{1}{2}}\left(e^{\ln\sigma_{x}-A^{\star}}\right)\tau^{-\frac{1}{2}}.\label{eq:optimal-measurement}
\end{equation}
Finally, the function
\begin{equation}
\Herm\ni A\mapsto\Tr\!\left[A\tau\right]+\sum_{x\in\mathcal{X}}\Tr\!\left[e^{\ln\sigma_{x}-A}\right]\label{eq:dual-objective-strictly-convex}
\end{equation}
is strictly convex.
\end{thm}

\begin{proof}
See Appendix~\ref{app:Proof-of-main-theorem}.
\end{proof}
To the best of our knowledge, the class of measurements $\left(M_{x}^{\star}\right)_{x\in\mathcal{X}}$
in~\eqref{eq:optimal-measurement}, in its most general form, has
not previously appeared in the quantum information literature. 

The
class of optimal measurements in~\eqref{eq:optimal-measurement} in Theorem~\ref{thm:main} unifies several
special cases of interest that we discuss further in Section~\ref{sec:special-cases}. In the classical case, the optimal measurement is independent of the choice of $\tau$, thus giving rise to a unique optimal measurement (see Corollary~\ref{cor:classical-case}). In the quantum case, on the other hand, there is no unique optimal measurement, and one can recover different optimal measurements that arise from different choices of $\tau$. 

\subsubsection{Computing the optimal measurement in Theorem~\ref{thm:main}}

\label{subsec:Computing-the-optimal-dual}

Here we devise an algorithm for finding the optimal value of the dual
objective in~\eqref{eq:dual-bayesian-min-change}. The value in doing
so is that, once one has the optimal $A^{\star}$, the optimal measurement
in~\eqref{eq:optimal-measurement} can be constructed directly.

Given that the function in~\eqref{eq:dual-objective-strictly-convex}
is strictly convex in $A$, one can employ gradient descent in order
to search for the optimum solution. This algorithm is simple, having
the following form:
\begin{equation}
A_{k+1}\leftarrow A_{k}-\eta\left(\tau-\sum_{x\in\mathcal{X}}e^{\ln\sigma_{x}-A_{k}}\right),
\end{equation}
where $k\in\mathbb{N}$ is an index for the iteration, $\eta\in\left(0,2\right)$
is the step size, and the matrix gradient $\frac{\partial}{\partial A}$
of~\eqref{eq:dual-objective-strictly-convex} is given by $\tau-\sum_{x\in\mathcal{X}}e^{\ln\sigma_{x}-A_{k}}$,
as shown in Appendix~\ref{app:Proof-of-main-theorem}.

This method has guaranteed local convergence near the optimum $A^{\star}$,
as established in Appendix~\ref{app:Local-convergence}. In practice,
one can combine this iteration with backtracking line search or trust-region
strategies to obtain a globally convergent descent method \cite[Chapters~3 \& 4]{NocedalWright2006}.

\section{Special cases of Theorem~\ref{thm:main}}

\label{sec:special-cases}

In this section, we consider several special cases of Theorem~\ref{thm:main},
which include the classical case, the pretty good measurement, a novel
measurement that we call the softmin thermal measurement, and the
Fermi--Dirac thermal measurement of~\cite{liu2026}.

\subsection{Classical case}

Let us begin by considering the case when the tuple $\left(\sigma_{x}\right)_{x\in\mathcal{X}}$,
as defined in~\eqref{eq:sigma-sub-state-def}, consists of commuting
substates and each substate $\sigma_{x}$ also commutes with the reference
state $\tau$. We can then write
\begin{align}
\sigma_{x} & =\sum_{y\in\mathcal{Y}}p(x)p(y|x)|y\rangle\!\langle y|,\label{eq:classical-case-commuting-rep}\\
\tau & =\sum_{y\in\mathcal{Y}}t(y)|y\rangle\!\langle y|,
\end{align}
where $\mathcal{Y}$ is an alphabet, $p(y|x)$ is a conditional probability
distribution, $t(y)$ is a probability distribution, and $\left\{ |y\rangle\right\} _{y\in\mathcal{Y}}$
is an orthonormal basis. In this case, the following corollary of
Theorem~\ref{thm:main} holds:
\begin{cor}
\label{cor:classical-case}Let $\tau$ be a positive definite reference
state, and let $\sigma_{x}>0$ be defined as in~\eqref{eq:sigma-sub-state-def}
for all $x\in\mathcal{X}$. Suppose that $\left[\tau,\sigma_{x}\right]=\left[\sigma_{x},\sigma_{x'}\right]=0$
for all $x,x'\in\mathcal{X}$. Then the following equality holds:
\begin{equation}
\min_{M_{\mathcal{X}}}D(Q_{\mathrm{rev}}(M_{\mathcal{X}})\|Q_{\mathrm{fwd}})=D(\tau\|\sigma),\label{eq:dual-bayesian-min-change-classical}
\end{equation}
where $M_{\mathcal{X}}$ is defined in~\eqref{eq:meas-shorthand},
$Q_{\mathrm{rev}}(M_{\mathcal{X}})$ is given by~\eqref{eq:q-rev-simple},
$Q_{\mathrm{fwd}}$ is given by~\eqref{eq:q-fwd-gen-1}, and $\sigma$
is defined in~\eqref{eq:average-sigma-state}. Furthermore, the unique
optimal $M_{x}^{\star}$ for all $x\in\mathcal{X}$ is given by
\begin{equation}
M_{x}^{\star}=\sigma_{x}\sigma^{-1}=\sum_{y\in\mathcal{Y}}p(x|y)|y\rangle\!\langle y|,\label{eq:optimal-meas-classical}
\end{equation}
where
\begin{align}
p(x|y) & \coloneqq\frac{p(y|x)p(x)}{p(y)},\\
p(y) & \coloneqq\sum_{x\in\mathcal{X}}p(y|x)p(x),
\end{align}
and $p(y|x)$ and $\left\{ |y\rangle\right\} _{y\in\mathcal{Y}}$
are defined in~\eqref{eq:classical-case-commuting-rep}.
\end{cor}

\begin{proof}
See Appendix~\ref{app:classical-case}.
\end{proof}
We note here that the finding in Corollary~\ref{cor:classical-case}
was already observed in~\cite{Bai2025}, where it was also noted that
the optimal reversal channel is independent of the reference state
$\tau$. Appendix~\ref{app:classical-case} provides an explicit proof
of these conclusions by starting from the dual formulation in~\eqref{eq:dual-bayesian-min-change}.

\subsubsection{Binary classical case}

\label{subsec:Binary-classical-case}

Here we briefly elaborate on the form of the Bayesian reversal channel
when the input alphabet $\mathcal{X}$ is binary. This further specializes
the classical case mentioned above and prepares for the developments
in Section~\ref{subsec:Sigmoid-neurons}.

Suppose that there is a binary random variable $X$ with a realization
$x\in\left\{ 0,1\right\} $ and probability $p\in\left(0,1\right)$
for $x=0$ and probability $1-p$ for $x=1$. Conditioned on the choice
of $x\in\left\{ 0,1\right\} $, a value $y\in\mathcal{Y}$ is chosen
according to the conditional probability
\begin{equation}
p(y|x)\equiv p_{Y|X}(y|x).
\end{equation}

The Bayes reversal channel $p_{X|Y}(0|y)$ in this scenario is given
by
\begin{align}
p_{X|Y}(0|y) & =\frac{p(y|0)p}{p(y|0)p+p(y|1)\left(1-p\right)}\\
 & =\frac{1}{1+\frac{p(y|1)\left(1-p\right)}{p(y|0)p}}\\
 & =\frac{1}{1+e^{\ln p(y|1)-\ln p(y|0)+\ln\left(\frac{1-p}{p}\right)}},\label{eq:Bayes-reversal-classical}
\end{align}
where we assumed that $p(y|0)>0$. Similarly, we find the following
for $p_{X|Y}(1|y)$:
\begin{align}
p_{X|Y}(1|y) & =\frac{1}{1+e^{-\left[\ln p(y|1)-\ln p(y|0)+\ln\left(\frac{1-p}{p}\right)\right]}},
\end{align}
and we observe that $p_{X|Y}(1|y)=1-p_{X|Y}(0|y)$, as expected.

\subsection{Pretty good measurement}

Here we consider the special case of Theorem~\ref{thm:main} when
the reference state $\tau=\sigma$, the latter defined in~\eqref{eq:average-sigma-state}.
In this case, we recover the pretty good measurement~\cite{Belavkin75,Belavkin75a,hughston1993complete,hausladen1994pretty}.
We note here that this finding was already observed in~\cite{Bai2025}.
\begin{cor}
\label{cor:pretty-good}Let $\tau$ be a positive definite reference
state such that $\tau=\sigma$, as defined in~\eqref{eq:average-sigma-state},
and let $\sigma_{x}>0$ be defined as in~\eqref{eq:sigma-sub-state-def}
for all $x\in\mathcal{X}$. Then the following equality holds:
\begin{equation}
\min_{M_{\mathcal{X}}}D(Q_{\mathrm{rev}}(M_{\mathcal{X}})\|Q_{\mathrm{fwd}})=0,\label{eq:dual-bayesian-min-change-classical-1}
\end{equation}
where $M_{\mathcal{X}}$ is defined in~\eqref{eq:meas-shorthand},
$Q_{\mathrm{rev}}(M_{\mathcal{X}})$ is given by~\eqref{eq:q-rev-simple},
and $Q_{\mathrm{fwd}}$ is given by~\eqref{eq:q-fwd-gen-1}. Furthermore,
the unique optimal $M_{x}^{\star}$ for all $x\in\mathcal{X}$ is
given by
\begin{equation}
M_{x}^{\star}=\sigma^{-\frac{1}{2}}\sigma_{x}\sigma^{-\frac{1}{2}}.\label{eq:PGM}
\end{equation}
\end{cor}

\begin{proof}
This is a direct consequence of Theorem~\ref{thm:main}, which we
can see in two different ways. First, by plugging the choice in~\eqref{eq:PGM}
into~\eqref{eq:rewrite-primal} and applying the assumption that $\tau=\sigma$,
we find that
\begin{align}
\sum_{x\in\mathcal{X}}D(\tau^{\frac{1}{2}}M_{x}^{\star}\tau^{\frac{1}{2}}\|\sigma_{x}) & =\sum_{x\in\mathcal{X}}D(\sigma^{\frac{1}{2}}\sigma^{-\frac{1}{2}}\sigma_{x}\sigma^{-\frac{1}{2}}\sigma^{\frac{1}{2}}\|\sigma_{x})\\
 & =\sum_{x\in\mathcal{X}}D(\sigma_{x}\|\sigma_{x})\\
 & =0,
\end{align}
thus achieving the minimum value. Alternatively, we can choose $A^{\star}=0$
for the dual problem in~\eqref{eq:dual-bayesian-min-change}, and
find that the dual objective reduces to
\begin{align}
 & 1-\left(\Tr\!\left[A^{\star}\tau\right]+\sum_{x\in\mathcal{X}}\Tr\!\left[e^{\ln\sigma_{x}-A^{\star}}\right]\right)\nonumber \\
 & =1-\sum_{x\in\mathcal{X}}\Tr\!\left[e^{\ln\sigma_{x}}\right]\\
 & =1-\sum_{x\in\mathcal{X}}\Tr\!\left[\sigma_{x}\right]\\
 & =0,
\end{align}
thus concluding the proof.
\end{proof}

\subsection{Softmin thermal measurement}

\label{subsec:Softmin-thermal-measurement}We consider a special case
of Theorem~\ref{thm:main} when the reference state $\tau=\frac{I}{d}$,
i.e., the maximally mixed state. In this case, the optimal measurement
is what we refer to as the softmin thermal measurement, as a generalization
of the softmin decision rule~\cite{Bridle1990,Bishop2006,Goodfellow2016}.
This nomenclature choice becomes more pertinent in Section~\ref{sec:Bayesian-minimal-change-thermal-states},
where we consider the form of this measurement for an ensemble of
thermal states.
\begin{cor}
\label{cor:softmax}Let $\tau=\frac{I}{d}$ be a reference state,
and let $\sigma_{x}>0$ be defined as in~\eqref{eq:sigma-sub-state-def}
for all $x\in\mathcal{X}$. Then the following equality holds:
\begin{multline}
\min_{M_{\mathcal{X}}}D(Q_{\mathrm{rev}}(M_{\mathcal{X}})\|Q_{\mathrm{fwd}})=\\
1-\inf_{A\in\Herm}\left\{ \frac{1}{d}\Tr\!\left[A\right]+\sum_{x\in\mathcal{X}}\Tr\!\left[e^{\ln\sigma_{x}-A}\right]\right\} ,\label{eq:dual-bayesian-min-change-1}
\end{multline}
where $M_{\mathcal{X}}$ is defined in~\eqref{eq:meas-shorthand},
$Q_{\mathrm{rev}}(M_{\mathcal{X}})$ is given by~\eqref{eq:q-rev-simple},
and $Q_{\mathrm{fwd}}$ is given by~\eqref{eq:q-fwd-gen-1}. Furthermore,
the dual optimization in~\eqref{eq:dual-bayesian-min-change-1} admits
a unique minimum $A^{\star}$ that satisfies
\begin{equation}
\sum_{x\in\mathcal{X}}e^{\ln\sigma_{x}-A^{\star}+\ln d}=I,\label{eq:constraint-optimal-povm-1-1}
\end{equation}
and the unique optimal $M_{x}^{\star}$ for all $x\in\mathcal{X}$
is given by
\begin{equation}
M_{x}^{\star}=e^{\ln\sigma_{x}-A^{\star}+\ln d}.\label{eq:q-softmax-meas}
\end{equation}
\end{cor}

\begin{proof}
This is an immediate corollary of Theorem~\ref{thm:main} that results
from setting $\tau=\frac{I}{d}$.
\end{proof}

\subsection{Fermi--Dirac thermal measurement}

Here we consider a special case of Corollary~\ref{cor:softmax} in
which $\left|\mathcal{X}\right|=2$; i.e., the alphabet is binary.
When doing so, we recover the Fermi--Dirac thermal measurement put
forward in~\cite{liu2026} and explored further in~\cite{he2026fermidiracmachines,he2026canonical}.
\begin{cor}
\label{cor:fermi-dirac}Let $\tau=\frac{I}{d}$ be a reference state,
and let $\sigma_{x}>0$ be defined as in~\eqref{eq:sigma-sub-state-def}
for all $x\in\mathcal{X}=\left\{ 0,1\right\} $. Then the following
equality holds:
\begin{multline}
\min_{M_{\mathcal{X}}}D(Q_{\mathrm{rev}}(M_{\mathcal{X}})\|Q_{\mathrm{fwd}})=\label{eq:fermi-dirac-obj}\\
D\!\left(\frac{M^{\star}}{d}\middle\|\sigma_{0}\right)+D\!\left(\frac{I-M^{\star}}{d}\middle\|\sigma_{1}\right),
\end{multline}
where $M_{\mathcal{X}}$ is defined in~\eqref{eq:meas-shorthand},
$Q_{\mathrm{rev}}(M_{\mathcal{X}})$ is given by~\eqref{eq:q-rev-simple},
and $Q_{\mathrm{fwd}}$ is given by~\eqref{eq:q-fwd-gen-1}. Furthermore,
the unique optimal $M^{\star}$ is given by
\begin{equation}
M^{\star}=\left(e^{\ln\sigma_{1}-\ln\sigma_{0}}+I\right)^{-1},\label{eq:FD-thermal-meas}
\end{equation}
so that
\begin{equation}
I-M^{\star}=\left(e^{-\left(\ln\sigma_{1}-\ln\sigma_{0}\right)}+I\right)^{-1}.
\end{equation}
\end{cor}

\begin{proof}
We provide a direct proof in Appendix~\ref{app:fermi-dirac}.
\end{proof}

\section{Relative-entropy minimum change principle for thermal states}

\label{sec:Bayesian-minimal-change-thermal-states}

In this section, we specialize Theorem~\ref{thm:main} and some of
the special cases in Section~\ref{sec:special-cases} to the case
when each state $\rho_{x}$ is a thermal state. One reason for doing
so is that thermal states play a key role in physics~\cite{Alhambra2023},
quantum optimization~\cite{Brandao2017,Brandao2019,Apeldoorn2019,vanApeldoorn2020quantumsdpsolvers,liu2025},
and quantum machine learning~\cite{Amin2018,Kieferova2017,Benedetti2017}.
When doing so, we recover softmin thermal and Fermi--Dirac thermal
measurements. The maximally mixed reference state yields the softmin
thermal measurement from Section~\ref{subsec:Softmin-thermal-measurement},
providing a multiclass quantum generalization of the classical softmin
rule. In the binary case, this further reduces to the Fermi--Dirac
thermal measurement from~\cite{liu2026}, a quantum generalization
of the sigmoid activation function explored further in~\cite{he2026fermidiracmachines,he2026canonical}
in the context of quantum machine learning.

In more detail, let $H_{x}$ be a Hamiltonian for all $x\in\mathcal{X}$.
Then the thermal state $\theta_{x}$ of this Hamiltonian is as follows:
\begin{align}
\theta_{x} & \coloneqq\frac{e^{-H_{x}}}{Z_{x}},\\
Z_{x} & \coloneqq\Tr\!\left[e^{-H_{x}}\right],
\end{align}
where we have absorbed the temperature into $H_{x}$ for simplicity
and brevity. Let $p(x)$ be a prior probability for the state $\theta_{x}$
and now define
\begin{equation}
\sigma_{x}\coloneqq p(x)\theta_{x}.\label{eq:sigma-substate-thermal}
\end{equation}
Consider that
\begin{align}
\ln\sigma_{x} & =-H_{x}',\\
\text{where }H_{x}' & \coloneqq H_{x}+\ln\!\left(\frac{Z_{x}}{p(x)}\right).
\end{align}

To begin with, observe that, for this case, the optimal measurement
from Theorem~\ref{thm:main} has the following measurement operators:
\begin{equation}
\tau^{-\frac{1}{2}}e^{-\left(H_{x}'+A^{\star}\right)}\tau^{-\frac{1}{2}},\label{eq:opt-meas-thermal}
\end{equation}
where $\tau>0$ is the reference state, $\sigma_{x}$ is chosen as
in~\eqref{eq:sigma-substate-thermal}, and each measurement operator
satisfies~\eqref{eq:constraint-optimal-povm-1}, so that the following
completeness relation holds:
\begin{equation}
\sum_{x\in\mathcal{X}}\tau^{-\frac{1}{2}}e^{-\left(H'_{x}+A^{\star}\right)}\tau^{-\frac{1}{2}}=I.
\end{equation}

\subsection{Softmin thermal measurements and Fermi--Dirac thermal measurements}

In the special case when $\tau=I/d$, i.e., the maximally mixed state,
Eq.~\eqref{eq:opt-meas-thermal} reduces to the softmin thermal measurement
of Corollary~\ref{cor:softmax}, with the following measurement operators:
\begin{equation}
e^{-\left(H_{x}''+A^{\star}\right)},\label{eq:softmax-thermal-meas}
\end{equation}
where $H''_{x}\coloneqq H'_{x}-\ln d$. Furthermore, the unique optimal
$A^{\star}$ satisfies
\begin{equation}
\sum_{x\in\mathcal{X}}e^{-\left(H_{x}''+A^{\star}\right)}=I,
\end{equation}
so that the measurement operators in~\eqref{eq:softmax-thermal-meas}
indeed form a legitimate POVM. The fact that each measurement operator
in~\eqref{eq:softmax-thermal-meas} consists of the exponential function
applied to $-\left(H_{x}''+A^{\star}\right)$ is the reason that we
call the measurement $\left(e^{-\left(H_{x}''+A^{\star}\right)}\right)_{x\in\mathcal{X}}$
a softmin thermal measurement. We elaborate further on this point
in Section~\ref{subsec:Exponential-families} below.

For the case of just two states, so that $\mathcal{X}=\left\{ 0,1\right\} $,
the optimal measurement from Theorem~\ref{thm:main} reduces to the
Fermi--Dirac thermal measurement, with the following two measurement
operators:
\begin{align}
M_{0} & =\left(e^{\Delta H}+I\right)^{-1},\label{eq:FD-thermal-M0}\\
M_{1} & =\left(e^{-\Delta H}+I\right)^{-1},\label{eq:FD-thermal-M1}
\end{align}
where
\begin{align}
\Delta H & \coloneqq H_{0}-H_{1}+\ln\!\left(\frac{Z_{0}\left(1-p\right)}{Z_{1}p}\right),
\end{align}
and we have set $p\equiv p(0)$, so that $1-p=p(1)$. This form of
the measurement is then in precise correspondence with the recent
proposal of~\cite{he2026fermidiracmachines,he2026canonical}, for
canonical quantization of neurons based on the sigmoid activation
function. We elaborate more on this point in Section~\ref{subsec:Sigmoid-neurons}.

\subsection{Classical case}

In this section, we consider the fully classical case, with our goal
being to justify how the softmin thermal measurement and Fermi--Dirac
thermal measurement represent quantum generalizations of Bayesian
reversal channels resulting from thermal probability distributions
(also called ``exponential family'').

\subsubsection{Exponential families}

\label{subsec:Exponential-families}Suppose that there is a prior
probability distribution $p_{X}(x)$ for the forward process, where
$x\in\mathcal{X}$. Suppose now that the forward channel is in the
exponential family:
\begin{align}
p_{Y|X}(y|x) & =\frac{e^{-H_{x}(y)}}{Z_{x}},\label{eq:exponential-family-dists}\\
Z_{x} & \coloneqq\sum_{y\in\mathcal{Y}}e^{-H_{x}(y)},
\end{align}
where $H_{x}(y)$ is a Hamiltonian function for $x\in\mathcal{X}$
(i.e., each $H_{x}$ is a polynomial function of the elements in $y$).
Then, after defining
\begin{equation}
H'_{x}(y)\coloneqq H_{x}(y)+\ln\!\left(\frac{Z_{x}}{p_{X}(x)}\right),
\end{equation}
it follows that the reverse channel has the following form:
\begin{align}
p_{X|Y}(x|y) & =\softmin(x),
\end{align}
where
\begin{equation}
\softmin(x)\coloneqq\frac{e^{-H'_{x}(y)}}{\sum_{x\in\mathcal{X}}e^{-H'_{x}(y)}}.
\end{equation}
The softmin function is widely used in machine learning as a method
for converting a vector of real numbers to a probability vector. It
is essentially a Boltzmann weighting scheme that assigns higher probability
to lower energy values and lower probability to higher energy values.
Thus, the softmin thermal measurement in~\eqref{eq:softmax-thermal-meas}
represents a quantum generalization of this concept.

Let us now consider the case when the alphabet $\mathcal{X}$ for
the forward process is binary; i.e., $\mathcal{X}=\left\{ 0,1\right\} $.
Continuing from the development in Section~\ref{subsec:Binary-classical-case},
observe for this case that
\begin{equation}
\ln p(y|1)-\ln p(y|0)=H_{0}(y)-H_{1}(y)+\ln\!\left(\frac{Z_{0}}{Z_{1}}\right),
\end{equation}
and plugging into~\eqref{eq:Bayes-reversal-classical} gives the following
for the Bayes reversal channel:
\begin{align}
p_{X|Y}(0|y) & =\frac{1}{1+e^{H(y)+b}},\label{eq:bayes-reverse-exp-fam-0}\\
p_{X|Y}(1|y) & =\frac{1}{1+e^{-\left[H(y)+b\right]}},\label{eq:bayes-reverse-exp-fam-1}\\
H(y) & \coloneqq H_{0}(y)-H_{1}(y),\\
b & \coloneqq\ln\!\left(\frac{Z_{0}\left(1-p\right)}{Z_{1}p}\right).
\end{align}
Thus, the conditional probabilities $p_{X|Y}(0|y)$ and $p_{X|Y}(1|y)$
have the form of a sigmoid function of $H(y)+b$, where we recall
that the sigmoid function has the following form:
\begin{equation}
z\mapsto\left(1+e^{z}\right)^{-1}.
\end{equation}
We now see that the Fermi--Dirac thermal measurement in~\eqref{eq:FD-thermal-M0}--\eqref{eq:FD-thermal-M1}
represents a quantum generalization of this concept.

\subsubsection{Sigmoid neurons and Boltzmann machines}

\label{subsec:Sigmoid-neurons}Suppose furthermore that each Hamiltonian
function is quadratic, so that
\begin{equation}
H_{x}(y)=y^{T}W_{x}y+w_{x}^{T}y,
\end{equation}
where $W_{x}\in\mathbb{R}^{n\times n}$ and $w_{x}\in\mathbb{R}^{n}$
for $x\in\left\{ 0,1\right\} $. In this case, the prior probability
densities in~\eqref{eq:exponential-family-dists} are known as Boltzmann
machines~\cite{Ackley1985,Hinton1986}. Then an immediate conclusion
from~\eqref{eq:bayes-reverse-exp-fam-0}--\eqref{eq:bayes-reverse-exp-fam-1}
is that
\begin{align}
p_{X|Y}(0|y) & =\frac{1}{1+e^{y^{T}Wy+w^{T}y+b}},\\
p_{X|Y}(1|y) & =\frac{1}{1+e^{-\left[y^{T}Wy+w^{T}y+b\right]}},\\
W & \coloneqq W_{0}-W_{1},\\
w & \coloneqq w_{0}-w_{1}.
\end{align}

The conditional probabilities again have the form of a sigmoid function,
but this time being a function of $y^{T}Wy+w^{T}y+b$. The above development
is a well known motivation for the sigmoid function used in classical
neurons.

\section{Softmin thermal measurements from semidefinite optimization}

\label{sec:Softmin-thermal-measurements-SDPs}

Theorem~\ref{thm:main} identifies the optimal measurement arising
from a relative-entropy minimum change principle. A natural question
is whether this same measurement also appears in other optimization
problems. In this section we show that the answer is affirmative.
In particular, we consider a broad class of entropy-regularized semidefinite
optimization problems whose optimal solutions are precisely the softmin
thermal measurements introduced in~\eqref{eq:softmax-thermal-meas}.
This provides an independent variational characterization of these
measurements and further highlights their role in quantum optimization. 

This result generalizes the optimization problem considered in~\cite{liu2026,Lindsey2023},
where the optimal solutions were shown to be Fermi--Dirac thermal
measurements. Allowing multiple measurement outcomes leads naturally
to the softmin thermal measurements of~\eqref{eq:softmax-thermal-meas}.
As such, the binary Fermi--Dirac thermal measurement is the two-outcome
member of a broader optimization framework.

Let us now specify the optimization problem. Let $c,d\in\mathbb{N}$,
let $\mathcal{X}$ be a finite alphabet, let $H_{x}$ be a $d\times d$
Hermitian matrix for all $x\in\mathcal{X}$, and let $Q_{i}$ be a
$d\times d$ Hermitian matrix for all $i\in\left[c\right]\equiv\left\{ 1,\ldots,c\right\} $.
Let $T>0$ be a temperature, and let $q_{x,i}\in\mathbb{R}$ for all
$x\in\mathcal{X}$ and $i\in\left[c\right]$. Adopting the notation
in~\eqref{eq:meas-shorthand}, the optimization problem is as follows:
\begin{equation}
\min_{M_{\mathcal{X}}}\left\{ 
\begin{array}{c}
     \sum_{x\in\mathcal{X}}\Tr\!\left[H_{x}M_{x}\right]-TS\!\left(M_{\mathcal{X}}\right):  \\
     \Tr\!\left[Q_{i}M_{x}\right]=q_{i,x}\,\forall i,x 
\end{array}\right\}
,\label{eq:measurement-free-energy-min}
\end{equation}
where the constraints hold for all $i\in\left[c\right]$ and $x\in\mathcal{X}$,
and the measurement entropy $S\!\left(M_{\mathcal{X}}\right)$ is
defined as
\begin{equation}
S\!\left(M_{\mathcal{X}}\right)\coloneqq-\sum_{x\in\mathcal{X}}\Tr\!\left[M_{x}\ln M_{x}\right].
\end{equation}
The optimization problem in~\eqref{eq:measurement-free-energy-min}
reduces to the following semidefinite program, close to the standard
form of \cite[Eq.~(4.51)]{Boyd2004}, in the limit $T\to0$:
\begin{equation}
\min_{M_{\mathcal{X}}}\left\{ \sum_{x\in\mathcal{X}}\Tr\!\left[H_{x}M_{x}\right]:\Tr\!\left[Q_{i}M_{x}\right]=q_{i,x}\forall i,x\right\} .
\end{equation}

The following theorem provides a complete dual characterization of
the optimization problem in~\eqref{eq:measurement-free-energy-min}
and shows that its optimal solutions are the softmin thermal measurements
from~\eqref{eq:softmax-thermal-meas}.
\begin{thm}
\label{thm:SDP-q-softmax}The following equality holds:
\begin{multline}
\min_{M_{\mathcal{X}}}\left\{ 
\begin{array}{c}
     \sum_{x\in\mathcal{X}}\Tr\!\left[H_{x}M_{x}\right]-TS\!\left(M_{\mathcal{X}}\right):  \\
     \Tr\!\left[Q_{i}M_{x}\right]=q_{i,x}\,\forall i,x 
\end{array}\right\} \\
=Td+\sup_{\vec{\mu}\in\mathbb{R}^{c\times\left|\mathcal{X}\right|}}\left\{ \vec{\mu}\cdot\vec{q}-f\!\left(\left(H_{x}\right)_{x},\vec{\mu},\left(Q_{i}\right)_{i}\right)\right\} ,
\end{multline}
where $\vec{\mu}\equiv\left(\mu_{i,x}\right)_{i,x}$, $\vec{\mu}\cdot\vec{q}\equiv\sum_{i,x}\mu_{i,x}q_{i,x}$,
and 
\begin{multline}
f\!\left(\left(H_{x}\right)_{x},\vec{\mu},\left(Q_{i}\right)_{i}\right)\coloneqq\\
\inf_{A\in\Herm}\left\{ \Tr[A]+T\sum_{x\in\mathcal{X}}\Tr\!\left[e^{-\frac{1}{T}\left(H_{x}-\sum_{i,x}\mu_{i,x}Q_{i}+A\right)}\right]\right\} .
\end{multline}
Furthermore, an optimal $A^{\star}$ for $f\!\left(\left(H_{x}\right)_{x},\vec{\mu},\left(Q_{i}\right)_{i}\right)$
satisfies
\begin{equation}
\sum_{x\in\mathcal{X}}e^{-\frac{1}{T}\left(H_{x}-\sum_{i,x}\mu_{i,x}Q_{i}+A^{\star}\right)}=I,
\end{equation}
and an optimal POVM consists of the measurement operators labeled
by $M_{x}^{\star}$, where
\begin{equation}
M_{x}^{\star}=e^{-\frac{1}{T}\left(H_{x}-\sum_{i,x}\mu_{i,x}Q_{i}+A^{\star}\right)}.
\end{equation}
\end{thm}

\begin{proof}
See Appendix~\ref{app:softmin-SDP}.
\end{proof}
Theorems~\ref{thm:main} and~\ref{thm:SDP-q-softmax} show that softmin
thermal measurements admit two distinct variational characterizations:
one as solutions of a relative-entropy minimum-change principle and
another as optimizers of entropy-regularized semidefinite programs.
This dual perspective suggests that softmin thermal measurements occupy
a role analogous to thermal states, the latter arising both from maximum-entropy
principles and from free-energy minimization.

\section{Connections to quantum information theory}

\label{sec:Connections-to-QIT}

In this section, we establish several links of the measurements in
Theorem~\ref{thm:main} to quantum information theory. First, we prove
that the minimum change principle in~\eqref{eq:minimal-change-I}
satisfies a particular additivity property (Section~\ref{subsec:Additivity-MCP}),
which is similar in spirit to the additivity of accessible information
\cite{Holevo1973StatisticalDecision}. After that, we establish bounds
on the performance of Fermi--Dirac thermal measurements in quantum
hypothesis testing (Section~\ref{subsec:Fermi=002013Dirac-thermal-measurements-QHT}).

\subsection{Additivity for the minimum change principle in Equation~\eqref{eq:minimal-change-I}}

\label{subsec:Additivity-MCP}

Here we state the additivity result, but before doing so, we establish
some notation. Let $\mathcal{E}_{1}$ and $\mathcal{E}_{2}$ be forward
channels of the form in~\eqref{eq:forward-channel}, so that
\begin{align}
\mathcal{E}_{1}(\omega_{1}) & \coloneqq\sum_{x_{1}\in\mathcal{X}_{1}}\langle x_{1}|\omega_{1}|x_{1}\rangle\rho_{x_{1}}^{(1)},\\
\mathcal{E}_{2}(\omega_{2}) & \coloneqq\sum_{x_{2}\in\mathcal{X}_{2}}\langle x_{2}|\omega_{2}|x_{2}\rangle\rho_{x_{2}}^{(2)}.
\end{align}
Also, suppose that the prior probability distribution for the product
channel $\mathcal{E}_{1}\otimes\mathcal{E}_{2}$ is a product distribution
$p_{X_{1}}(x_{1})p_{X_{2}}(x_{2})$, and suppose furthermore that
the reference state for the reverse process is a product state $\tau_{1}\otimes\tau_{2}$.
Then we use the following notations for the various bipartite states
involved in the forward and reverse processes, for both the individual
and combined cases:
\begin{align}
Q_{\mathrm{fwd}}(p_{X_{1}},\mathcal{E}_{1}) & \coloneqq\sum_{x_{1}\in\mathcal{X}_{1}}\sigma_{x_{1}}^{(1)}\otimes|x_{1}\rangle\!\langle x_{1}|,\label{eq:additive-notation-1}\\
Q_{\mathrm{fwd}}(p_{X_{2}},\mathcal{E}_{2}) & \coloneqq\sum_{x_{2}\in\mathcal{X}_{2}}\sigma_{x_{2}}^{(2)}\otimes|x_{2}\rangle\!\langle x_{2}|,\\
\sigma_{x_{1}}^{(1)} & \coloneqq p_{X_{1}}(x_{1})\rho_{x_{1}}^{(1)},\\
\sigma_{x_{2}}^{(2)} & \coloneqq p_{X_{2}}(x_{2})\rho_{x_{2}}^{(2)},\\
Q_{\mathrm{rev}}(\tau_{1},M_{\mathcal{X}_{1}}) & \coloneqq\sum_{x_{1}\in\mathcal{X}_{1}}\tau_{1}^{\frac{1}{2}}M_{x_{1}}^{(1)}\tau_{1}^{\frac{1}{2}}\otimes|x_{1}\rangle\!\langle x_{1}|,\\
Q_{\mathrm{rev}}(\tau_{2},M_{\mathcal{X}_{2}}) & \coloneqq\sum_{x_{2}\in\mathcal{X}_{2}}\tau_{2}^{\frac{1}{2}}M_{x_{2}}^{(2)}\tau_{2}^{\frac{1}{2}}\otimes|x_{2}\rangle\!\langle x_{2}|,
\end{align}
\begin{multline}
Q_{\mathrm{fwd}}(p_{X_{1}}\otimes p_{X_{2}},\mathcal{E}_{1}\otimes\mathcal{E}_{2})\coloneqq\\
\sum_{\substack{x_{1}\in\mathcal{X}_{1},\\
x_{2}\in\mathcal{X}_{2}
}
}\sigma_{x_{1}}^{(1)}\otimes\sigma_{x_{2}}^{(2)}\otimes|x_{1}\rangle\!\langle x_{1}|\otimes|x_{2}\rangle\!\langle x_{2}|,
\end{multline}
\begin{multline}
Q_{\mathrm{rev}}(\tau_{1}\otimes\tau_{2},M_{\mathcal{X}_{1}\times\mathcal{X}_{2}})\coloneqq\\
\sum_{\substack{x_{1}\in\mathcal{X}_{1},\\
x_{2}\in\mathcal{X}_{2}
}
}\left(\tau_{1}\otimes\tau_{2}\right)^{\frac{1}{2}}M_{x_{1}x_{2}}\left(\tau_{1}\otimes\tau_{2}\right)^{\frac{1}{2}}\otimes|x_{1}\rangle\!\langle x_{1}|\otimes|x_{2}\rangle\!\langle x_{2}|.\label{eq:additive-notation-last}
\end{multline}

\begin{thm}
\label{thm:additivity}Given the notation in~\eqref{eq:additive-notation-1}--\eqref{eq:additive-notation-last},
the following additivity relation holds:
\begin{multline}
\!\!\!\!\!\!\min_{M_{\mathcal{X}_{1}\times\mathcal{X}_{2}}}D(Q_{\mathrm{rev}}(\tau_{1}\otimes\tau_{2},M_{\mathcal{X}_{1}\times\mathcal{X}_{2}})\|Q_{\mathrm{fwd}}(p_{X_{1}}\otimes p_{X_{2}},\mathcal{E}_{1}\otimes\mathcal{E}_{2}))\\
=\sum_{i=1}^{2}\min_{M_{\mathcal{X}_{i}}}D(Q_{\mathrm{rev}}(\tau_{i},M_{\mathcal{X}_{i}})\|Q_{\mathrm{fwd}}(p_{X_{i}},\mathcal{E}_{i})).
\end{multline}
\end{thm}

\begin{proof}
See Appendix~\ref{app:additivity}.
\end{proof}
An important implication of Theorem~\ref{thm:additivity} is that,
for a forward process consisting of a product of two processes and
a reference state that is a product of two reference states, the optimal
reversal channel is a product of the reversal channels that are optimal
for each individual case.

\subsection{Fermi--Dirac thermal measurements and quantum hypothesis testing}

\label{subsec:Fermi=002013Dirac-thermal-measurements-QHT}

In this section, we investigate the performance of the Fermi--Dirac
thermal measurement in~\eqref{eq:FD-thermal-meas} for quantum hypothesis
testing. Our main results consist of upper bounds on the error probability
when using this measurement in both the one-shot and asymptotic scenarios.
We also conclude that, in the asymptotic independent and identically
distributed (i.i.d.) scenario, this measurement consists of a product
measurement followed by classical postprocessing, implying that it
cannot achieve optimal asymptotic performance in general. One conclusion
of this finding is that there cannot exist a multiplicative constant
relating the error probability of the Fermi--Dirac thermal measurement
to that of the optimal measurement, as there does for the pretty good
measurement~\cite{BarnumKnill2002}.

\subsubsection{Review of quantum hypothesis testing}

Let us begin by reviewing some well known results in quantum hypothesis
testing. In the scenario of binary symmetric hypothesis testing, a
state $\rho_{0}$ is prepared with probability $p_{0}$ and a state
$\rho_{1}$ is prepared with probability $p_{1}$, where $p_{0}+p_{1}=1$.
For simplicity, we assume throughout that $\rho_{0}$ and $\rho_{1}$
are positive definite states and $p_{0},p_{1}>0$. The optimal error
probability in deciding which state is prepared is as follows~\cite{Helstrom1967,Helstrom1969,Holevo1972}:
\begin{align}
p_{\mathrm{Hel}}^{e} & \coloneqq\min_{M:0\leq M\leq I}\left\{ \Tr\!\left[M\sigma_{0}\right]+\Tr\!\left[\left(I-M\right)\sigma_{1}\right]\right\} \label{eq:helstrom-error-prob}\\
 & =\frac{1}{2}\left(1-\left\Vert \Delta\right\Vert _{1}\right),
\end{align}
where
\begin{equation}
\Delta\coloneqq\sigma_{1}-\sigma_{0},
\end{equation}
and we have adopted the same notation from~\eqref{eq:sigma-sub-state-def},
so that
\begin{equation}
\sigma_{x}\coloneqq p_{x}\rho_{x}\quad\forall x\in\left\{ 0,1\right\} .
\end{equation}
The optimal measurement operator $M^{\star}$ in~\eqref{eq:helstrom-error-prob}
is known as the Helstrom measurement, equal to the projection $\Pi_{+}$
onto the positive eigenspace of $\Delta$~\cite{Helstrom1967,Helstrom1969}.
Letting $\Pi_{-}$ denote the projection onto the orthogonal subspace,
we can write
\begin{equation}
p_{\mathrm{Hel}}^{e}=\frac{1}{2}\left(1-\Tr\!\left[\signum(\Delta)\Delta\right]\right),
\end{equation}
where
\begin{equation}
\signum(\Delta)\coloneqq\Pi_{+}-\Pi_{-}.
\end{equation}
This follows because $\left\Vert \Delta\right\Vert _{1}=\Tr\!\left[\signum(\Delta)\Delta\right]$.

Let $p_{\mathrm{PG}}^{e}$ denote the error probability of the pretty
good measurement:
\begin{equation}
p_{\mathrm{PG}}^{e}=\Tr\!\left[\sigma^{-\frac{1}{2}}\sigma_{1}\sigma^{-\frac{1}{2}}\sigma_{0}\right]+\Tr\!\left[\sigma^{-\frac{1}{2}}\sigma_{0}\sigma^{-\frac{1}{2}}\sigma_{1}\right],
\end{equation}
where
\begin{equation}
\sigma\coloneqq\sigma_{0}+\sigma_{1}.
\end{equation}
The following inequalities hold
\begin{equation}
p_{\mathrm{Hel}}^{e}\leq p_{\mathrm{PG}}^{e}\leq2p_{\mathrm{Hel}}^{e},\label{eq:PG-to-opt}
\end{equation}
where the second inequality was proven in~\cite{BarnumKnill2002}
(see also \cite[Theorem~3.10]{watrous2018theory}). The upper bound
in~\eqref{eq:PG-to-opt} justifies the name ``pretty good,'' and
we note that this result was generalized recently in~\cite{Mishra2025}.

\subsubsection{Non-asymptotic performance of Fermi--Dirac thermal measurements}

We now establish upper bounds on the error probability in hypothesis
testing when using Fermi--Dirac thermal measurements. Defining
\begin{equation}
A\coloneqq\ln\sigma_{1}-\ln\sigma_{0},
\end{equation}
we can write the Fermi--Dirac thermal measurement in~\eqref{eq:FD-thermal-meas}
as follows:
\begin{equation}
\left(\left(e^{A}+I\right)^{-1},\left(e^{-A}+I\right)^{-1}\right),\label{eq:FD-thermal-meas-A}
\end{equation}
where the first outcome is identified with guessing ``0'' and the
second with guessing ``1''. Then we can write the error probability,
when using this measurement for symmetric hypothesis testing, as follows:
\begin{align}
p_{\mathrm{FD}}^{e} & \coloneqq\Tr\!\left[\left(e^{-A}+I\right)^{-1}\sigma_{0}\right]+\Tr\!\left[\left(e^{A}+I\right)^{-1}\sigma_{1}\right]\label{eq:FD-error-prob-QHT}\\
 & =\frac{1}{2}\left(1-\Tr\!\left[\tanh\!\left(A/2\right)\Delta\right]\right),\label{eq:FD-err-prob-tanh}
\end{align}
where the second equality follows from simple manipulations and is
proved in Appendix~\ref{app:Proofs-FD-QHT}.

The following inequality is immediate, a consequence of the fact that
the Helstrom measurement is optimal:
\begin{equation}
p_{\mathrm{Hel}}^{e}\leq p_{\mathrm{FD}}^{e}.\label{eq:Hel-to-FD}
\end{equation}
A key mathematical difference between the optimal error probability
and that in~\eqref{eq:FD-err-prob-tanh} is the substitution of $\signum(\Delta)$
with $\tanh\!\left(A/2\right)$. The function $\tanh$ is a smooth
approximation of the sgn function, and in the commuting (classical)
case, i.e., when $\left[\sigma_{0},\sigma_{1}\right]=0$, the operators
$\signum(\Delta)$ and $\signum(A)$ are the same. However, in the
general noncommuting case, there is a strong distinction between the
operators $\signum(\Delta)$ and $\signum(A)$, making it difficult
to relate the performance of the Fermi--Dirac thermal measurement
to that of the optimal measurement, by means of a relation other than
that in~\eqref{eq:Hel-to-FD}.

We now progress towards establishing an upper bound on the error probability
$p_{\mathrm{FD}}^{e}$. Before doing so, let us recall the Chernoff
divergence between states $\rho_{0}$ and $\rho_{1}$, defined as
\begin{align}
C(\rho_{0}\|\rho_{1}) & \coloneqq\sup_{s\in\left(0,1\right)}C_{s}(\rho_{0}\|\rho_{1}),\label{eq:chernoff-div}\\
C_{s}(\rho_{0}\|\rho_{1}) & \coloneqq-\ln\Tr\!\left[\rho_{0}^{s}\rho_{1}^{1-s}\right].\label{eq:unopt-chernoff}
\end{align}
It is the optimal error exponent for asymptotic quantum hypothesis
testing~\cite{Audenaert2007,Nussbaum2009} (discussed more in Section
\ref{subsec:Asymptotic-performance-FD-QHT}). Here, we define the
following quantity that is equivalent to~\eqref{eq:unopt-chernoff}
whenever $\rho_{0}$ and $\rho_{1}$ commute but is different otherwise
(see Lemma~\ref{lem:usual-chernoff-to-modified}):
\begin{equation}
C_{s}^{\natural}(\rho_{0}\|\rho_{1})\coloneqq-\ln\Tr\!\left[e^{\left(1-s\right)\left(\ln\rho_{1}-\ln\rho_{0}\right)}\rho_{0}\right].\label{eq:unopt-unusal-chernoff}
\end{equation}
This quantity is related to the performance of the Fermi--Dirac thermal
measurement in quantum hypothesis testing, as stated in the following
lemma:
\begin{lem}
\label{lem:one-shot-FD-err-bnd}The error probability in one-shot
symmetric quantum hypothesis testing, when using a Fermi--Dirac thermal
measurement of the form in~\eqref{eq:FD-thermal-meas}, is bounded
as follows for all $s\in\left[0,1\right]$:
\begin{equation}
p_{\mathrm{FD}}^{e}\leq2g(s,p_{0})e^{-\min\left\{ C_{s}^{\natural}(\rho_{0}\|\rho_{1}),C_{1-s}^{\natural}(\rho_{1}\|\rho_{0})\right\} },\label{eq:err-prob-bnd-QHT-FD-unusal-Chernoff}
\end{equation}
where $p_{\mathrm{FD}}^{e}$ is defined in~\eqref{eq:FD-error-prob-QHT},
$C_{s}^{\natural}(\rho_{0}\|\rho_{1})$ is defined in~\eqref{eq:unopt-unusal-chernoff},
and 
\begin{equation}
g(s,p_{0})\coloneqq\left(sp_{0}\right)^{s}\left(\left(1-s\right)p_{1}\right)^{1-s}.\label{eq:g-s-func}
\end{equation}
\end{lem}

The following lemma relates the usual Chernoff divergence to that
appearing in Lemma~\ref{lem:one-shot-FD-err-bnd}:
\begin{lem}
\label{lem:usual-chernoff-to-modified}The following inequality holds
for all positive definite states $\rho_{0}$ and $\rho_{1}$ and $s\in\left(0,1\right)$:
\begin{equation}
C_{s}(\rho_{0}\|\rho_{1})\geq\min\!\left\{ C_{s}^{\natural}(\rho_{0}\|\rho_{1}),C_{1-s}^{\natural}(\rho_{1}\|\rho_{0})\right\} ,\label{eq:usual-chernoff-to-modified}
\end{equation}
and the inequality is strict if $\rho_{0}$ and $\rho_{1}$ do not
commute.
\end{lem}

\begin{proof}
See Appendix~\ref{app:usual-chernoff-to-modified}.
\end{proof}
Recall the following bound from \cite[Theorem~1]{Audenaert2007}:
\begin{equation}
p_{\mathrm{Hel}}^{e}\leq p_{0}^{s}p_{1}^{1-s}e^{-C(\rho_{0}\|\rho_{1})}.\label{eq:q-chernoff-bnd}
\end{equation}
Ignoring prefactors, Lemma~\ref{lem:usual-chernoff-to-modified} indicates
that the bound in~\eqref{eq:q-chernoff-bnd} is generally a stronger
upper bound on the error probability of symmetric hypothesis testing
than is the bound in~\eqref{eq:err-prob-bnd-QHT-FD-unusal-Chernoff},
and the distinction between them becomes more clear in the asymptotic
scenario discussed in the next section.

\subsubsection{Asymptotic performance of Fermi--Dirac thermal measurements}

\label{subsec:Asymptotic-performance-FD-QHT}

In this section, we discuss the performance of Fermi--Dirac thermal
measurements in asymptotic hypothesis testing and establish an upper
bound on the error probability when doing so. We also mention how
Fermi--Dirac thermal measurements for i.i.d.~states can be realized
as a product measurement followed by classical postprocessing.

To begin with, note that the optimal error probability for symmetric
hypothesis testing, when given access to $n$ samples of the unknown
state, is equal to
\begin{equation}
p_{\mathrm{Hel},n}^{e}\coloneqq\frac{1}{2}\left(1-\left\Vert p_{1}\rho_{1}^{\otimes n}-p_{0}\rho_{0}^{\otimes n}\right\Vert _{1}\right),
\end{equation}
which follows simply by substituting $\rho_{x}$ with $\rho_{x}^{\otimes n}$
in~\eqref{eq:helstrom-error-prob}, for all $x\in\left\{ 0,1\right\} $.
Following the same reasoning after~\eqref{eq:helstrom-error-prob},
an optimal measurement operator achieving this error probability is
given by the projection onto the positive eigenspace of $p_{1}\rho_{1}^{\otimes n}-p_{0}\rho_{0}^{\otimes n}$.
In general, implementing the corresponding measurement requires a
collective strategy, which cannot be realized by means of a product
measurement followed by classical postprocessing. The optimal asymptotic
error exponent is given by the quantum Chernoff divergence~\cite{Audenaert2007,Nussbaum2009}:
\begin{equation}
\lim_{n\to\infty}-\frac{1}{n}\ln p_{\mathrm{Hel},n}^{e}=C(\rho_{0}\|\rho_{1}),
\end{equation}
where $C(\rho_{0}\|\rho_{1})$ is defined in~\eqref{eq:chernoff-div}.

Let us now consider the performance of the Fermi--Dirac thermal measurement
for symmetric hypothesis testing. Defining
\begin{align}
A^{(n)} & \coloneqq\ln\!\left(p_{1}\rho_{1}^{\otimes n}\right)-\ln\!\left(p_{0}\rho_{0}^{\otimes n}\right)\\
 & =\ln\!\left(\frac{p_{1}}{p_{0}}\right)+\sum_{i=1}^{n}\left(\ln\rho_{1}^{(i)}-\ln\rho_{0}^{(i)}\right),\label{eq:additive-form-A-n}
\end{align}
this measurement is as follows:
\begin{equation}
\left(\left(e^{A^{(n)}}+I^{\otimes n}\right)^{-1},\left(e^{-A^{(n)}}+I^{\otimes n}\right)^{-1}\right),\label{eq:n-fold-FD-meas}
\end{equation}
which follows by making the substitution $A\to A^{(n)}$ in~\eqref{eq:FD-thermal-meas-A}.
Then the error probability, when using this measurement for symmetric
hypothesis testing, is as follows:
\begin{multline}
p_{\mathrm{FD},n}^{e}\coloneqq\Tr\!\left[\left(e^{-A^{(n)}}+I\right)^{-1}\left(p_{0}\rho_{0}^{\otimes n}\right)\right]\\
+\Tr\!\left[\left(e^{A^{(n)}}+I\right)^{-1}\left(p_{1}\rho_{1}^{\otimes n}\right)\right].\label{eq:FD-finite-n-err-prob}
\end{multline}

We can then make the following conclusion when using the Fermi--Dirac
measurement in asymptotic quantum hypothesis testing:
\begin{prop}
\label{prop:FD-asymptotic-QHT}Fix $n\in\mathbb{N}$. The error probability
in $n$-shot symmetric quantum hypothesis testing, when using the
measurement in~\eqref{eq:n-fold-FD-meas}, is bounded as follows for
all $s\in\left[0,1\right]$:
\begin{multline}
p_{\mathrm{FD},n}^{e}\leq2g(s,p_{0})e^{-n\min\left\{ C_{s}^{\natural}(\rho_{0}\|\rho_{1}),C_{1-s}^{\natural}(\rho_{1}\|\rho_{0})\right\} },
\end{multline}
where $p_{\mathrm{FD},n}^{e}$ is defined in~\eqref{eq:FD-finite-n-err-prob},
$g(s,p_{0})$ in~\eqref{eq:g-s-func}, and $C_{s}^{\natural}(\rho_{0}\|\rho_{1})$
in~\eqref{eq:unopt-unusal-chernoff}.
\end{prop}

\begin{proof}
This is a direct consequence of Lemmas~\ref{lem:one-shot-FD-err-bnd}
and~\ref{lem:additivity-unopt-chernoff}, the latter establishing
additivity of $C_{s}^{\natural}(\rho_{0}\|\rho_{1})$. See Appendix
\ref{app:FD-asymptotic-QHT}.
\end{proof}
A direct consequence of the bound in Proposition~\ref{prop:FD-asymptotic-QHT}
is as follows:
\begin{cor}
When using Fermi--Dirac thermal measurements, the following lower
bound on the error exponent holds:
\begin{multline}
\lim_{n\to\infty}-\frac{1}{n}\ln p_{\mathrm{FD},n}^{e}\geq\\
\sup_{s\in\left(0,1\right)}\min\left\{ C_{s}^{\natural}(\rho_{0}\|\rho_{1}),C_{1-s}^{\natural}(\rho_{1}\|\rho_{0})\right\} ,
\end{multline}
where $p_{\mathrm{FD},n}^{e}$ is defined in~\eqref{eq:FD-finite-n-err-prob}
and $C_{s}^{\natural}(\rho_{0}\|\rho_{1})$ in~\eqref{eq:unopt-unusal-chernoff}.
\end{cor}

As a consequence of the form of the operator $A^{(n)}$ in~\eqref{eq:additive-form-A-n},
it follows that the measurement operators in~\eqref{eq:n-fold-FD-meas}
can be written in the following way, respectively:
\begin{align}
 & \sum_{j^{n}}\left(e^{a(j^{n})}+1\right)^{-1}|\phi_{j_{1}}\rangle\!\langle\phi_{j_{1}}|\otimes\cdots\otimes|\phi_{j_{n}}\rangle\!\langle\phi_{j_{n}}|,\label{eq:FD-n-fold-0}\\
 & \sum_{j^{n}}\left(e^{-a(j^{n})}+1\right)^{-1}|\phi_{j_{1}}\rangle\!\langle\phi_{j_{1}}|\otimes\cdots\otimes|\phi_{j_{n}}\rangle\!\langle\phi_{j_{n}}|,\label{eq:FD-n-fold-1}
\end{align}
where
\begin{align}
j^{n} & \equiv\left(j_{1},\ldots,j_{n}\right),\\
a(j^{n}) & \equiv\ln\!\left(\frac{p_{1}}{p_{0}}\right)+\sum_{i=1}^{n}h_{j_{i}},
\end{align}
and a spectral decomposition of $\ln\rho_{1}-\ln\rho_{0}$ is given
by
\begin{equation}
\ln\rho_{1}-\ln\rho_{0}=\sum_{j}h_{j}|\phi_{j}\rangle\!\langle\phi_{j}|.\label{eq:eigenbasis-for-FD}
\end{equation}

\begin{rem}
\label{rem:FD-asymp-decomp}The decomposition of the measurement operators
in~\eqref{eq:FD-n-fold-0}--\eqref{eq:FD-n-fold-1} implies that
the Fermi--Dirac thermal measurement in~\eqref{eq:n-fold-FD-meas}
can be realized by means of a product measurement followed by classical
postprocessing. See Appendix~\ref{app:Implementing-Fermi=002013Dirac-thermal-iid}
for more details.
\end{rem}

Based on the observation in Remark~\ref{rem:FD-asymp-decomp}, we
conclude the following upper bound on the error exponent when using
a Fermi--Dirac thermal measurement:
\begin{cor}
\label{cor:up-bnd-exp-FD}When using Fermi--Dirac thermal measurements,
the following upper bound on the error exponent holds:
\begin{equation}
\lim_{n\to\infty}-\frac{1}{n}\ln p_{\mathrm{FD},n}^{e}\leq C(p_{0}\|p_{1}),
\end{equation}
where $p_{\mathrm{FD},n}^{e}$ is defined in~\eqref{eq:FD-finite-n-err-prob},
$C(p_{0}\|p_{1})$ from~\eqref{eq:chernoff-div}, and $p_{0}$ and
$p_{1}$ are the following commuting (classical) states resulting
from measuring $\rho_{0}$ and $\rho_{1}$ in the eigenbasis defined
in~\eqref{eq:eigenbasis-for-FD}:
\begin{align}
p_{0} & \coloneqq\sum_{j}\langle\phi_{j}|\rho_{0}|\phi_{j}\rangle\,|\phi_{j}\rangle\!\langle\phi_{j}|,\\
p_{1} & \coloneqq\sum_{j}\langle\phi_{j}|\rho_{1}|\phi_{j}\rangle\,|\phi_{j}\rangle\!\langle\phi_{j}|.
\end{align}
\end{cor}

\begin{proof}
This follows from Remark~\ref{rem:FD-asymp-decomp} and the optimal
error exponent for the classical case~\cite{Chernoff1952Measure}.
Indeed, the Fermi--Dirac thermal measurement begins with a product
measurement in the eigenbasis defined in~\eqref{eq:eigenbasis-for-FD}.
This induces classical product distributions, which are then subject
to the optimal limits from the classical case.
\end{proof}
We note here that, if desired, one can arrive at a non-asymptotic
upper bound on the error exponent by making use of \cite[Remark~16]{Ji2026}.
\begin{rem}
\label{rem:FD-not-PG}We finally conclude that the upper bound in
Corollary~\ref{cor:up-bnd-exp-FD} is strictly less than the quantum
Chernoff divergence $C(\rho_{0}\|\rho_{1})$ if and only if the states
do not commute, a direct consequence of Lemma~\ref{lem:strict-ineq-petz-measured}
in Appendix~\ref{app:usual-chernoff-to-modified}. As such, although
one might hope for a general bound similar to that in~\eqref{eq:PG-to-opt},
but for the Fermi--Dirac thermal measurement instead of the pretty
good measurement, this is impossible. For if there were, then the
Fermi--Dirac thermal measurement would achieve the optimal quantum
Chernoff exponent. However, Corollary~\ref{cor:up-bnd-exp-FD} excludes
this possibility.
\end{rem}

\section{Conclusion}

\label{sec:Conclusion}

In summary, following the framework of~\cite{Bai2025}, our main contribution
is to construct a solution -- i.e., an optimal measurement -- for a minimum change principle for quantum relative entropy
in the setting of quantum statistical inference (Theorem~\ref{thm:main}). 
Here, the forward process is classical to quantum (typically called
``preparation''), while the reverse process is quantum to classical
(measurement). By doing so, we found that the optimal measurement
has a novel form, as given in~\eqref{eq:optimal-measurement}, and
reduces to known measurements and a novel softmin thermal measurement
in~\eqref{eq:q-softmax-meas}. We then considered this measurement
for the special cases of thermal states, linking it to known classical
decision rules such as softmin and sigmoid. We also showed how the
softmin thermal measurement arises in a different context of entropic-regularized
semidefinite optimization, extending our earlier developments in~\cite{liu2026}.
Finally, we connected the minimum change principle in Theorem~\ref{thm:main}
to quantum information theory in two different ways, establishing
an additivity property for it and exploring the performance of Fermi--Dirac
thermal measurements in quantum hypothesis testing.

Going forward, we suspect that the softmin thermal measurement will
have applications in multiclassification tasks in quantum machine
learning, similar to how the Fermi--Dirac thermal measurement was
shown to have applications in binary classification~\cite{he2026fermidiracmachines,he2026canonical}.
However, a key obstacle to overcome is to determine how to implement
them as a quantum algorithm. This was accomplished in~\cite{liu2026}
and explored further in~\cite{he2026fermidiracmachines,he2026canonical},
but it remains unclear how to do so due to the presence of the operator
$A^{\star}$ in~\eqref{eq:q-softmax-meas}.

This minimum change principle for quantum relative entropy demonstrates how well-known measurements like pretty good measurements and other measurements arise from different choices of $\tau$, suggesting a variational interpretation for $\tau$. Thus, an important question is whether other choices of $\tau$ can lead to other meaningful measurements.

A particularly natural interpolation is to choose the reference state $\tau=e^{-\beta H}/\operatorname{Tr}[e^{-\beta H}]$, for some Hamiltonian $H$. For instance, when $[H, \sigma_x]=[\sigma_x, \sigma_{x'}]=0$ for all $x, x'$, then the optimal measurement reproduces the classical Bayes theorem (Corollary~\ref{cor:classical-case}). With the choice $H=-\ln \sigma$, then $\tau = \sigma^{\beta}/\operatorname{Tr}[\sigma^\beta]$, where we think of $\beta \geq 0 $ as an inverse temperature parameter that interpolates between the maximally mixed state at $\beta = 0$ (infinite temperature) and the state $\sigma$ in~\eqref{eq:average-sigma-state} at $\beta = 1$. The resulting measurement in~\eqref{eq:optimal-meas} then interpolates between the softmin thermal measurement at $\beta = 0$ (Corollary~\ref{cor:softmax}) and the pretty good measurement at $\beta = 1$ (Corollary~\ref{cor:pretty-good}). In this context, it is interesting to determine the performance of the resulting measurement in symmetric hypothesis testing, in comparison to the optimal measurement. As noted in~\eqref{eq:PG-to-opt}, the error probability of the pretty good measurement is no larger than twice that of the optimal measurement, and we wonder how this bound generalizes as a function of $\beta$, for the aforementioned interpolation choice. A future analysis on how $M^\star$ is modified due to small deviations in $\tau$ could also add insight into the landscape of different optimal measurements that is interpolated by $\tau$.

Finally, we suspect that our additivity result in Theorem~\ref{thm:additivity} has implications for an operational interpretation of the minimum change principle in~\eqref{eq:minimal-change-I} in a context of quantum hypothesis testing different from that already presented in Section~\ref{subsec:Asymptotic-performance-FD-QHT}. Namely, we think it should be relevant in the context of the quantum Sanov theorem~\cite{Bjelakovic2005QuantumSanov,Bjelakovic2008TypicalSupport, Noetzel2014QuantumSanov,Lami2026Asymptotic,Lami2025GeneralisedSanov}, in which the null hypothesis has a non-i.i.d.~structure but the alternative hypothesis is i.i.d.~(see especially \cite[Theorem~14]{Lami2026Asymptotic} here). In this setting, we expect for our additivity result to imply that the optimal error exponent simplifies considerably (i.e., be ``single-letter'' in the parlance of information theory). We leave this direction of inquiry open for future work.

\medskip{}

\textit{Note on independent work}---After completing the results
in Section~\ref{subsec:Fermi=002013Dirac-thermal-measurements-QHT}
of our paper, we noticed~\cite{Meunson2026} posted to the arXiv,
in which the quantity $C_{s}^{\natural}(\rho_{0}\|\rho_{1})$ was
defined and some of its properties established.

\section*{Acknowledgements}

We thank Ludovico Lami for a helpful discussion related to Remark
\ref{rem:FD-not-PG}. We also thank Mil\'an Mosonyi for pointing us
to \cite[Theorem~4.18]{Hiai2017Different} and \cite[Remark~III.12]{Hiai2023TestMeasured}.

NL acknowledges funding from the Science and Technology Commission
of Shanghai Municipality (STCSM) grant no.~24LZ1401200 (21JC1402900),
NSFC grants no.~12471411 and no.~12341104, the Shanghai Jiao Tong
University 2030 Initiative, the Shanghai Pilot Program for Basic Research,
and the Fundamental Research Funds for the Central Universities. MMW
acknowledges support from the National Science Foundation under grant
nos.~2329662 and 2611810.

% \bibliographystyle{unsrturl}
% \phantomsection\addcontentsline{toc}{section}{\refname}

\bibliography{Ref}

\appendix
\onecolumngrid
\large

\section{Alternative minimum change principle in Equation~\eqref{eq:minimal-change-II}}

\label{app:Alternative-minimum-change}

In the main text, we focused on the minimum change principle based
on minimizing $D(Q_{\mathrm{rev}}(M_{\mathcal{X}})\|Q_{\text{fwd}})$,
because it admits a simple dual formulation and an explicit characterization
of the optimal measurement. In this appendix, we briefly examine the
alternative ordering $D(Q_{\text{fwd}}\|Q_{\mathrm{rev}}(M_{\mathcal{X}}))$.
Although this optimization is well motivated and recovers the pretty
good measurement in an important special case, it appears to have
a substantially more complicated mathematical structure.

For convenience, let us recall Eq.~\eqref{eq:minimal-change-II}
here:
\begin{equation}
\min_{M_{\mathcal{X}}}D(Q_{\text{fwd}}\|Q_{\mathrm{rev}}(M_{\mathcal{X}}))=\min_{M_{\mathcal{X}}}\sum_{x\in\mathcal{X}}D\!\left(\sigma_{x}\middle\|\tau^{\frac{1}{2}}M_{x}\tau^{\frac{1}{2}}\right).\label{eq:recall-alt-MCP}
\end{equation}
To begin with, let us note that, when the reference state $\tau=\sigma$,
as defined in~\eqref{eq:sigma-sub-state-def}, the optimal measurement
is the pretty good measurement. This was already noted in~\cite{Bai2025},
and the proof is the same as the first part of the proof of Corollary
\ref{cor:pretty-good}. It was also noted in~\cite{Bai2025} that,
in the commuting case, the optimal reversal channel is the Bayes reversal
(we omit proving that explicitly here).

In general, since the relative entropy is convex in its arguments,
it follows that one can use relative entropy optimization methods
\cite{FawziSaunderson2023Optimal,He2025,He2025a,KossmannSchwonnek2026Optimising}
to find the optimal value of~\eqref{eq:minimal-change-II}, as well
as an optimal measurement. 

However, it is less clear whether one can solve this problem in the
way that we have done so in Theorem~\ref{thm:main}. To see this point,
let us consider a simple case in which the forward channel features
just two states. Then the optimization in~\eqref{eq:recall-alt-MCP}
reduces to
\begin{align}
 & \min_{M:0\leq M\leq I}\left\{ D\!\left(\sigma_{0}\middle\|\tau^{\frac{1}{2}}M\tau^{\frac{1}{2}}\right)+D\!\left(\sigma_{1}\middle\|\tau^{\frac{1}{2}}\left(I-M\right)\tau^{\frac{1}{2}}\right)\right\} \nonumber \\
 & =\min_{M:0\leq M\leq I}\left\{ -S(\sigma_{0})-\Tr\!\left[\sigma_{0}\ln\!\left(\tau^{\frac{1}{2}}M\tau^{\frac{1}{2}}\right)\right]-S(\sigma_{1})-\Tr\!\left[\sigma_{1}\ln\!\left(\tau^{\frac{1}{2}}\left(I-M\right)\tau^{\frac{1}{2}}\right)\right]\right\} \\
 & =-S(\sigma_{0})-S(\sigma_{1})-\max_{M:0\leq M\leq I}\left\{ \Tr\!\left[\sigma_{0}\ln\!\left(\tau^{\frac{1}{2}}M\tau^{\frac{1}{2}}\right)\right]+\Tr\!\left[\sigma_{1}\ln\!\left(\tau^{\frac{1}{2}}\left(I-M\right)\tau^{\frac{1}{2}}\right)\right]\right\} ,
\end{align}
where $S(\omega)\coloneqq-\Tr\!\left[\omega\ln\omega\right]$. For
simplicity, let us just set $\tau=I/d$. Then
\begin{align}
\Tr\!\left[\sigma_{0}\ln\!\left(\tau^{\frac{1}{2}}M\tau^{\frac{1}{2}}\right)\right]+\Tr\!\left[\sigma_{1}\ln\!\left(\tau^{\frac{1}{2}}\left(I-M\right)\tau^{\frac{1}{2}}\right)\right] & =\Tr\!\left[\sigma_{0}\ln\!\left(M/d\right)\right]+\Tr\!\left[\sigma_{1}\ln\!\left(\left(I-M\right)/d\right)\right]\\
 & =\Tr\!\left[\sigma_{0}\ln M\right]+\Tr\!\left[\sigma_{1}\ln\!\left(I-M\right)\right]-\ln d.
\end{align}
Taking the matrix derivative of the last line with respect to $M$,
we find that
\begin{equation}
\frac{\partial}{\partial M}\left(\Tr\!\left[\sigma_{0}\ln M\right]+\Tr\!\left[\sigma_{1}\ln\!\left(I-M\right)\right]\right)=D\ln(M)[\sigma_{0}]-D\ln(I-M)[\sigma_{1}],
\end{equation}
where the notation $D\ln(M)[\sigma_{0}]$ and $D\ln(I-M)[\sigma_{1}]$
denotes the Fr\'echet derivative. Setting this matrix gradient equal
to zero gives the following nonlinear equation in $M$:
\begin{equation}
D\ln(M)[\sigma_{0}]=D\ln(I-M)[\sigma_{1}].
\end{equation}
This is equivalent to
\begin{equation}
\int_{0}^{\infty}ds\left(M+sI\right)^{-1}\sigma_{0}\left(M+sI\right)^{-1}=\int_{0}^{\infty}ds\left(I-M+sI\right)^{-1}\sigma_{1}\left(I-M+sI\right)^{-1},
\end{equation}
and it is unclear how to solve this equation for $M$ in general.
We remark here that the issue is similar to that encountered when
trying to optimize the measured relative entropy of states, as discussed
in \cite[Section~2.1]{Sreekumar2026performance}.

\section{Proof of Theorem~\ref{thm:main}}

\label{app:Proof-of-main-theorem}

Let $F_{x}\geq0$ and $G_{x}>0$ for all $x\in\mathcal{X}$. The direct-sum
property of the quantum relative entropy is as follows:
\begin{equation}
D\!\left(\sum_{x\in\mathcal{X}}F_{x}\otimes|x\rangle\!\langle x|\middle\|\sum_{x\in\mathcal{X}}G_{x}\otimes|x\rangle\!\langle x|\right)=\sum_{x\in\mathcal{X}}D\!\left(F_{x}\middle\|G_{x}\right).\label{eq:direct-sum-rel-ent}
\end{equation}
By employing~\eqref{eq:direct-sum-rel-ent}, we find that the minimum
change principles in~\eqref{eq:minimal-change-I} and~\eqref{eq:minimal-change-II}
can be rewritten as follows:
\begin{align}
D(Q_{\mathrm{rev}}(M_{\mathcal{X}})\|Q_{\text{fwd}}) & =D\!\left(\sum_{x\in\mathcal{X}}\tau^{\frac{1}{2}}M_{x}\tau^{\frac{1}{2}}\otimes|x\rangle\!\langle x|\middle\|\sum_{x\in\mathcal{X}}\sigma_{x}\otimes|x\rangle\!\langle x|\right)\\
 & =\sum_{x\in\mathcal{X}}D\!\left(\tau^{\frac{1}{2}}M_{x}\tau^{\frac{1}{2}}\middle\|\sigma_{x}\right),
\end{align}
as well as
\begin{equation}
D(Q_{\text{fwd}}\|Q_{\mathrm{rev}}(M_{\mathcal{X}}))=\sum_{x\in\mathcal{X}}D\!\left(\sigma_{x}\middle\|\tau^{\frac{1}{2}}M_{x}\tau^{\frac{1}{2}}\right).
\end{equation}
Under the assumption that $\tau>0$, we can perform the substitution
\begin{equation}
N_{x}=\tau^{\frac{1}{2}}M_{x}\tau^{\frac{1}{2}}\label{eq:Nx-sub}
\end{equation}
for all $x\in\mathcal{X}$ and observe that
\begin{equation}
\sum_{x\in\mathcal{X}}M_{x}=I\quad\Leftrightarrow\quad\sum_{x\in\mathcal{X}}N_{x}=\tau
\end{equation}
to rewrite the optimization problems in~\eqref{eq:minimal-change-I}
and~\eqref{eq:minimal-change-II} as follows:
\begin{align}
\min_{M_{\mathcal{X}}}D(Q_{\mathrm{rev}}(M_{\mathcal{X}})\|Q_{\text{fwd}}) & =\min_{N_{x}\geq0\,\forall x}\left\{ \sum_{x\in\mathcal{X}}D\!\left(N_{x}\middle\|\sigma_{x}\right):\sum_{x\in\mathcal{X}}N_{x}=\tau\right\} ,\label{eq:minimal-change-opt-I}\\
\min_{M_{\mathcal{X}}}D(Q_{\text{fwd}}\|Q_{\mathrm{rev}}(M_{\mathcal{X}})) & =\min_{N_{x}\geq0\,\forall x}\left\{ \sum_{x\in\mathcal{X}}D\!\left(\sigma_{x}\middle\|N_{x}\right):\sum_{x\in\mathcal{X}}N_{x}=\tau\right\} .
\end{align}

Let us define the generalized relative entropy of $F\geq0$ and $G>0$
as 
\begin{equation}
\widetilde{D}\!\left(F\middle\|G\right)\coloneqq\Tr\!\left[F\left(\ln F-\ln G\right)\right]+\Tr\!\left[G\right]-\Tr\!\left[F\right],\label{eq:gen-rel-ent}
\end{equation}
which has the following faithfulness property \cite[Appendix~B]{Falk1970}:
\begin{equation}
\widetilde{D}\!\left(F\middle\|G\right)=0\qquad\Leftrightarrow\qquad F=G.
\end{equation}

Consider that
\begin{align}
 & \min_{\substack{N_{x}\geq0\,\forall x,\\
\sum_{x}N_{x}=\tau
}
}\sum_{x\in\mathcal{X}}D\!\left(N_{x}\middle\|\sigma_{x}\right)\nonumber \\
 & =\min_{\substack{N_{x}\geq0\,\forall x,\\
\sum_{x}N_{x}=\tau
}
}\sum_{x\in\mathcal{X}}\widetilde{D}\!\left(N_{x}\middle\|\sigma_{x}\right)\\
 & =\min_{N_{x}\geq0\,\forall x}\left\{ \sum_{x\in\mathcal{X}}\widetilde{D}\!\left(N_{x}\middle\|\sigma_{x}\right)+\sup_{A\in\Herm}\Tr\!\left[A\left(\sum_{x\in\mathcal{X}}N_{x}-\tau\right)\right]\right\} \\
 & =\min_{N_{x}\geq0\,\forall x}\sup_{A\in\Herm}\left\{ -\Tr\!\left[A\tau\right]+\sum_{x\in\mathcal{X}}\widetilde{D}\!\left(N_{x}\middle\|\sigma_{x}\right)+\Tr\!\left[AN_{x}\right]\right\} \\
 & =\sup_{A\in\Herm}\min_{N_{x}\geq0\,\forall x}\left\{ -\Tr\!\left[A\tau\right]+\sum_{x\in\mathcal{X}}\widetilde{D}\!\left(N_{x}\middle\|\sigma_{x}\right)+\Tr\!\left[AN_{x}\right]\right\} \label{eq:minimax-step}\\
 & =\sup_{A\in\Herm}\min_{N_{x}\geq0\,\forall x}\left\{ -\Tr\!\left[A\tau\right]+\sum_{x\in\mathcal{X}}\Tr\!\left[N_{x}\left(\ln N_{x}-\left(\ln\sigma_{x}-A\right)\right)\right]+\Tr\!\left[\sigma_{x}-N_{x}\right]\right\} \\
 & =\sup_{A\in\Herm}\min_{N_{x}\geq0\,\forall x}\left\{ 1-\Tr\!\left[A\tau\right]+\sum_{x\in\mathcal{X}}\Tr\!\left[N_{x}\left(\ln N_{x}-\ln e^{\ln\sigma_{x}-A}\right)\right]-\Tr\!\left[N_{x}\right]\right\} \\
 & =\sup_{A\in\Herm}\min_{N_{x}\geq0\,\forall x}\left\{ 1-\Tr\!\left[A\tau\right]+\sum_{x\in\mathcal{X}}\widetilde{D}(N_{x}\|e^{\ln\sigma_{x}-A})-\Tr\!\left[e^{\ln\sigma_{x}-A}\right]\right\} \label{eq:gen-rel-ent-step}\\
 & =\sup_{A\in\Herm}\left\{ 1-\Tr\!\left[A\tau\right]-\sum_{x\in\mathcal{X}}\Tr\!\left[e^{\ln\sigma_{x}-A}\right]+\min_{N_{x}\geq0\,\forall x}\sum_{x\in\mathcal{X}}\widetilde{D}(N_{x}\|e^{\ln\sigma_{x}-A})\right\} \\
 & =\sup_{A\in\Herm}\left\{ 1-\Tr\!\left[A\tau\right]-\sum_{x\in\mathcal{X}}\Tr\!\left[e^{\ln\sigma_{x}-A}\right]\right\} .
\end{align}
The first equality follows because
\begin{equation}
D\!\left(N_{x}\middle\|\sigma_{x}\right)+\Tr\!\left[\sigma_{x}-N_{x}\right]=\widetilde{D}\!\left(N_{x}\middle\|\sigma_{x}\right)
\end{equation}
and $\sum_{x\in\mathcal{X}}\Tr\!\left[\sigma_{x}-N_{x}\right]=0$,
given that
\begin{align}
\sum_{x\in\mathcal{X}}\Tr\!\left[\sigma_{x}\right] & =1,\\
\sum_{x\in\mathcal{X}}\Tr\!\left[N_{x}\right] & =\Tr[\tau]=1.
\end{align}
The minimax equality in~\eqref{eq:minimax-step} follows because the
objective function is convex in $N_{x}$ and linear in $A$. It is
also a consequence of Slater's theorem, which holds because the primal
feasible set contains interior points such as $N_{x}=\tau/\left|\mathcal{X}\right|$.
The equality in~\eqref{eq:gen-rel-ent-step} follows from the identity
\begin{equation}
\Tr\!\left[N_{x}\left(\ln N_{x}-\ln e^{\ln\sigma_{x}-A}\right)\right]+\Tr\!\left[e^{\ln\sigma_{x}-A}\right]-\Tr\!\left[N_{x}\right]=\widetilde{D}(N_{x}\|e^{\ln\sigma_{x}-A}).
\end{equation}
The last equality follows from applying the faithfulness of $\widetilde{D}(N_{x}\|e^{\ln\sigma_{x}-A})$.

Observe that
\begin{equation}
\sup_{A\in\Herm}\left\{ 1-\Tr\!\left[A\tau\right]-\sum_{x\in\mathcal{X}}\Tr\!\left[e^{\ln\sigma_{x}-A}\right]\right\} =1-\inf_{A\in\Herm}\left\{ \Tr\!\left[A\tau\right]+\sum_{x\in\mathcal{X}}\Tr\!\left[e^{\ln\sigma_{x}-A}\right]\right\} ,
\end{equation}
thus establishing~\eqref{eq:dual-bayesian-min-change}.

The dual stationarity condition is that
\begin{align}
0 & =\frac{\partial}{\partial A}\left(1-\Tr\!\left[A\tau\right]-\sum_{x\in\mathcal{X}}\Tr\!\left[e^{\ln\sigma_{x}-A}\right]\right)\\
 & =-\tau+\sum_{x\in\mathcal{X}}e^{\ln\sigma_{x}-A},\label{eq:stationarity-cond}
\end{align}
which follows because
\begin{equation}
\frac{\partial}{\partial A}\Tr\!\left[e^{\ln\sigma_{x}-A}\right]=-e^{\ln\sigma_{x}-A}.\label{eq:matrix-deriv-exp}
\end{equation}
To see the equality in~\eqref{eq:matrix-deriv-exp}, define the matrix
elements of $A$ as $a_{ij}\equiv\left[A\right]_{ij}$ and consider
that
\begin{align}
\frac{\partial}{\partial a_{ji}}\Tr\!\left[e^{\ln\sigma_{x}-A}\right] & =\frac{\partial}{\partial a_{ji}}\Tr\!\left[\sum_{n=0}^{\infty}\frac{\left(\ln\sigma_{x}-A\right)^{n}}{n!}\right]\\
 & =\frac{\partial}{\partial a_{ji}}\Tr\!\left[\sum_{n=1}^{\infty}\frac{\left(\ln\sigma_{x}-A\right)^{n}}{n!}\right]\\
 & =\sum_{n=1}^{\infty}\frac{1}{n!}\Tr\!\left[\frac{\partial}{\partial a_{ji}}\left(\ln\sigma_{x}-A\right)^{n}\right]\\
 & =\sum_{n=1}^{\infty}\frac{1}{n!}\left(\begin{array}{c}
\Tr\!\left[\left(\frac{\partial}{\partial a_{ji}}\left(\ln\sigma_{x}-A\right)\right)\left(\ln\sigma_{x}-A\right)^{n-1}\right]+\\
\Tr\!\left[\left(\ln\sigma_{x}-A\right)\left(\frac{\partial}{\partial a_{ji}}\left(\ln\sigma_{x}-A\right)\right)\left(\ln\sigma_{x}-A\right)^{n-2}\right]+\cdots+\\
\Tr\!\left[\left(\ln\sigma_{x}-A\right)^{n-1}\left(\frac{\partial}{\partial a_{ji}}\left(\ln\sigma_{x}-A\right)\right)\right]
\end{array}\right)\\
 & =-\sum_{n=1}^{\infty}\frac{1}{n!}\left(\begin{array}{c}
\Tr\!\left[|j\rangle\langle i|\left(\ln\sigma_{x}-A\right)^{n-1}\right]+\\
\Tr\!\left[\left(\ln\sigma_{x}-A\right)|j\rangle\langle i|\left(\ln\sigma_{x}-A\right)^{n-2}\right]+\cdots+\\
\Tr\!\left[\left(\ln\sigma_{x}-A\right)^{n-1}|j\rangle\langle i|\right]
\end{array}\right)\\
 & =-\sum_{n=1}^{\infty}\frac{n}{n!}\langle i|\left(\ln\sigma_{x}-A\right)^{n-1}|j\rangle\\
 & =-\langle i|\left(\sum_{n=1}^{\infty}\frac{\left(\ln\sigma_{x}-A\right)^{n-1}}{n-1!}\right)|j\rangle\\
 & =-\langle i|\left(\sum_{n=0}^{\infty}\frac{\left(\ln\sigma_{x}-A\right)^{n}}{n!}\right)|j\rangle\\
 & =-\langle i|e^{\ln\sigma_{x}-A}|j\rangle.
\end{align}
The stationarity condition in~\eqref{eq:stationarity-cond} implies
that the following equality holds for the optimal $A^{\star}$:
\begin{equation}
\sum_{x\in\mathcal{X}}e^{\ln\sigma_{x}-A^{\star}}=\tau.\label{eq:constraint-optimal-povm}
\end{equation}

The optimal $A^{\star}$ is unique because the objective function
\begin{equation}
A\mapsto1-\Tr\!\left[A\tau\right]-\sum_{x\in\mathcal{X}}\Tr\!\left[e^{\ln\sigma_{x}-A}\right]
\end{equation}
is strictly concave in $A$. To see this, let us compute the second
partial derivatives of the objective function, which are given by
\begin{align}
\frac{\partial}{\partial a_{\ell k}}\frac{\partial}{\partial a_{ji}}\left(1-\Tr\!\left[A\tau\right]-\sum_{x\in\mathcal{X}}\Tr\!\left[e^{\ln\sigma_{x}-A}\right]\right) & =\frac{\partial}{\partial a_{\ell k}}\left(-\langle i|\sum_{x\in\mathcal{X}}e^{\ln\sigma_{x}-A}|j\rangle\right)\\
 & =-\sum_{x\in\mathcal{X}}\langle i|\left(\frac{\partial}{\partial a_{\ell k}}e^{\ln\sigma_{x}-A}\right)|j\rangle.
\end{align}
We find that
\begin{align}
\frac{\partial}{\partial a_{\ell k}}e^{\ln\sigma_{x}-A} & =\frac{\partial}{\partial a_{\ell k}}\sum_{n=0}^{\infty}\frac{\left(\ln\sigma_{x}-A\right)^{n}}{n!}\\
 & =\frac{\partial}{\partial a_{\ell k}}\sum_{n=1}^{\infty}\frac{\left(\ln\sigma_{x}-A\right)^{n}}{n!}\\
 & =\sum_{n=1}^{\infty}\frac{1}{n!}\frac{\partial}{\partial a_{\ell k}}\left(\ln\sigma_{x}-A\right)^{n}\\
 & =\sum_{n=1}^{\infty}\frac{1}{n!}\left(\begin{array}{c}
\left(\frac{\partial}{\partial a_{\ell k}}\left(\ln\sigma_{x}-A\right)\right)\left(\ln\sigma_{x}-A\right)^{n-1}+\\
\left(\ln\sigma_{x}-A\right)\left(\frac{\partial}{\partial a_{\ell k}}\left(\ln\sigma_{x}-A\right)\right)\left(\ln\sigma_{x}-A\right)^{n-2}+\cdots+\\
\left(\ln\sigma_{x}-A\right)^{n-1}\left(\frac{\partial}{\partial a_{\ell k}}\left(\ln\sigma_{x}-A\right)\right)
\end{array}\right)\\
 & =-\sum_{n=1}^{\infty}\frac{1}{n!}\left(\begin{array}{c}
|\ell\rangle\langle k|\left(\ln\sigma_{x}-A\right)^{n-1}+\\
\left(\ln\sigma_{x}-A\right)|\ell\rangle\langle k|\left(\ln\sigma_{x}-A\right)^{n-2}+\cdots+\\
\left(\ln\sigma_{x}-A\right)^{n-1}|\ell\rangle\langle k|.
\end{array}\right)\\
 & =-\sum_{n=1}^{\infty}\frac{1}{n!}\sum_{m=0}^{n-1}\left(\ln\sigma_{x}-A\right)^{m}|\ell\rangle\langle k|\left(\ln\sigma_{x}-A\right)^{n-m-1}.
\end{align}
Now let $\sum_{r}\lambda_{r}\Pi_{r}$ be the spectral decomposition
of $\ln\sigma_{x}-A$:
\begin{equation}
\ln\sigma_{x}-A=\sum_{r}\lambda_{r}\Pi_{r}.
\end{equation}
Then
\begin{align}
 & -\sum_{n=1}^{\infty}\frac{1}{n!}\sum_{m=0}^{n-1}\left(\ln\sigma_{x}-A\right)^{m}|\ell\rangle\langle k|\left(\ln\sigma_{x}-A\right)^{n-m-1}\nonumber \\
 & =-\sum_{n=1}^{\infty}\frac{1}{n!}\sum_{m=0}^{n-1}\left(\sum_{r}\lambda_{r}\Pi_{r}\right)^{m}|\ell\rangle\langle k|\left(\sum_{r'}\lambda_{r'}\Pi_{r'}\right)^{n-m-1}\\
 & =-\sum_{n=1}^{\infty}\frac{1}{n!}\sum_{m=0}^{n-1}\left(\sum_{r}\lambda_{r}^{m}\Pi_{r}\right)|\ell\rangle\langle k|\left(\sum_{r'}\lambda_{r'}^{n-m-1}\Pi_{r'}\right)\\
 & =-\sum_{n=1}^{\infty}\frac{1}{n!}\sum_{r,r'}\left(\sum_{m=0}^{n-1}\lambda_{r}^{m}\lambda_{r'}^{n-m-1}\right)\Pi_{r}|\ell\rangle\langle k|\Pi_{r'}\\
 & =-\sum_{n=1}^{\infty}\frac{1}{n!}\left(\sum_{r}n\lambda_{r}^{n-1}\Pi_{r}|\ell\rangle\langle k|\Pi_{r}+\sum_{r,r'}\frac{\lambda_{r}^{n}-\lambda_{r'}^{n}}{\lambda_{r}-\lambda_{r'}}\Pi_{r}|\ell\rangle\langle k|\Pi_{r'}\right)\\
 & =-\left(\sum_{r}\left(\sum_{n=1}^{\infty}\frac{1}{n!}n\lambda_{r}^{n-1}\right)\Pi_{r}|\ell\rangle\langle k|\Pi_{r}+\sum_{r,r'}\frac{\sum_{n=1}^{\infty}\frac{1}{n!}\left(\lambda_{r}^{n}-\lambda_{r'}^{n}\right)}{\lambda_{r}-\lambda_{r'}}\Pi_{r}|\ell\rangle\langle k|\Pi_{r'}\right)\\
 & =-\left(\sum_{r}e^{\lambda_{r}}\Pi_{r}|\ell\rangle\langle k|\Pi_{r}+\sum_{r,r'}\frac{\sum_{n=0}^{\infty}\frac{1}{n!}\left(\lambda_{r}^{n}-\lambda_{r'}^{n}\right)}{\lambda_{r}-\lambda_{r'}}\Pi_{r}|\ell\rangle\langle k|\Pi_{r'}\right)\\
 & =-\left(\sum_{r}e^{\lambda_{r}}\Pi_{r}|\ell\rangle\langle k|\Pi_{r}+\sum_{r,r'}\frac{e^{\lambda_{r}}-e^{\lambda_{r'}}}{\lambda_{r}-\lambda_{r'}}\Pi_{r}|\ell\rangle\langle k|\Pi_{r'}\right)\\
 & =-\int_{0}^{1}ds\,e^{s\left(\ln\sigma_{x}-A\right)}|\ell\rangle\langle k|e^{\left(1-s\right)\left(\ln\sigma_{x}-A\right)}.
\end{align}
The last equality follows from the proof of \cite[Proposition~47]{wilde2025fisher}.
This finally implies that the second partial derivatives are as follows:
\begin{equation}
\frac{\partial}{\partial a_{\ell k}}\frac{\partial}{\partial a_{ji}}\left(1-\Tr\!\left[A\tau\right]-\sum_{x\in\mathcal{X}}\Tr\!\left[e^{\ln\sigma_{x}-A}\right]\right)=-\langle i|\sum_{x\in\mathcal{X}}\int_{0}^{1}ds\,e^{s\left(\ln\sigma_{x}-A\right)}|\ell\rangle\langle k|e^{\left(1-s\right)\left(\ln\sigma_{x}-A\right)}|j\rangle,
\end{equation}
so that the Hessian superoperator for the objective function is given
by
\begin{equation}
\mathcal{H}(\omega)\coloneqq-\sum_{x\in\mathcal{X}}\int_{0}^{1}ds\,e^{s\left(\ln\sigma_{x}-A\right)}\omega e^{\left(1-s\right)\left(\ln\sigma_{x}-A\right)}.\label{eq:Hessian-superoperator}
\end{equation}
Towards establishing strict concavity of the objective function, define
\begin{equation}
K_{x}\equiv\ln\sigma_{x}-A
\end{equation}
and consider that
\begin{align}
\max_{\omega:\left\Vert \omega\right\Vert _{2}=1}\left\langle \omega,\mathcal{H}(\omega)\right\rangle  & =\max_{\omega:\left\Vert \omega\right\Vert _{2}=1}-\left\langle \omega,\sum_{x\in\mathcal{X}}\int_{0}^{1}ds\,e^{sK_{x}}\omega e^{\left(1-s\right)K_{x}}\right\rangle \\
 & =-\min_{\omega:\left\Vert \omega\right\Vert _{2}=1}\sum_{x\in\mathcal{X}}\int_{0}^{1}ds\,\left\langle \omega,e^{sK_{x}}\omega e^{\left(1-s\right)K_{x}}\right\rangle \\
 & =-\min_{\omega:\left\Vert \omega\right\Vert _{2}=1}\sum_{x\in\mathcal{X}}\int_{0}^{1}ds\,\Tr\!\left[\omega^{\dag}e^{sK_{x}}\omega e^{\left(1-s\right)K_{x}}\right]\\
 & =-\min_{\omega:\left\Vert \omega\right\Vert _{2}=1}\sum_{x\in\mathcal{X}}\int_{0}^{1}ds\,\Tr\!\left[e^{\left(1-s\right)K_{x}/2}\omega^{\dag}e^{sK_{x}/2}e^{sK_{x}/2}\omega e^{\left(1-s\right)K_{x}/2}\right]\\
 & =-\min_{\omega:\left\Vert \omega\right\Vert _{2}=1}\sum_{x\in\mathcal{X}}\int_{0}^{1}ds\,\Tr\!\left[\left(e^{sK_{x}/2}\omega e^{\left(1-s\right)K_{x}/2}\right)^{\dag}e^{sK_{x}/2}\omega e^{\left(1-s\right)K_{x}/2}\right]\\
 & =-\min_{\omega:\left\Vert \omega\right\Vert _{2}=1}\sum_{x\in\mathcal{X}}\int_{0}^{1}ds\,\left\Vert e^{sK_{x}/2}\omega e^{\left(1-s\right)K_{x}/2}\right\Vert _{2}^{2}\\
 & \leq0.
\end{align}
The last inequality follows because
\begin{equation}
\sum_{x\in\mathcal{X}}\int_{0}^{1}ds\,\left\Vert e^{sK_{x}/2}\omega e^{\left(1-s\right)K_{x}/2}\right\Vert _{2}^{2}\geq0,
\end{equation}
given that $\left\Vert e^{sK_{x}/2}\omega e^{\left(1-s\right)K_{x}/2}\right\Vert _{2}^{2}\geq0$
for all $x\in\mathcal{X}$ and $s\in\left[0,1\right]$. We can in
fact prove that
\begin{equation}
\max_{\omega:\left\Vert \omega\right\Vert _{2}=1}\left\langle \omega,\mathcal{H}(\omega)\right\rangle <0,
\end{equation}
thus establishing strict concavity, by employing a proof by contradiction.
Suppose that
\begin{equation}
\max_{\omega:\left\Vert \omega\right\Vert _{2}=1}\left\langle \omega,\mathcal{H}(\omega)\right\rangle =0.
\end{equation}
Then this implies that there exists $\omega$ such that $\left\Vert \omega\right\Vert _{2}=1$
and
\begin{equation}
\sum_{x\in\mathcal{X}}\int_{0}^{1}ds\,\left\Vert e^{sK_{x}/2}\omega e^{\left(1-s\right)K_{x}/2}\right\Vert _{2}^{2}=0.
\end{equation}
This further implies that
\begin{equation}
\left\Vert e^{sK_{x}/2}\omega e^{\left(1-s\right)K_{x}/2}\right\Vert _{2}^{2}=0
\end{equation}
for all $s\in\left[0,1\right]$ and $x\in\mathcal{X}$, which in turn
implies that $e^{sK_{x}/2}\omega e^{\left(1-s\right)K_{x}/2}=0$.
Since $e^{sK_{x}/2}$ and $e^{\left(1-s\right)K_{x}/2}$ are invertible
for all $s\in\left[0,1\right]$ and $x\in\mathcal{X}$, we conclude
that $\omega=0$, thus contradicting the assumption that $\left\Vert \omega\right\Vert _{2}=1$.

Finally, from strict concavity, we conclude that the optimal $N_{x}^{\star}$
for all $x\in\mathcal{X}$ in the primal problem is unique and given
by
\begin{equation}
N_{x}^{\star}=e^{\ln\sigma_{x}-A^{\star}}
\end{equation}
where $A^{\star}$ satisfies~\eqref{eq:constraint-optimal-povm}.
The uniqueness property holds from the faithfulness of the generalized
relative entropy in~\eqref{eq:gen-rel-ent}. By the substitution in
\eqref{eq:Nx-sub}, we then conclude that
\begin{equation}
M_{x}^{\star}=\tau^{-\frac{1}{2}}\left(e^{\ln\sigma_{x}-A^{\star}}\right)\tau^{-\frac{1}{2}},
\end{equation}
as claimed.

\section{Local convergence of gradient descent algorithm}

\label{app:Local-convergence}

Here we analyze an upper bound on the maximum eigenvalue of the Hessian
superoperator for~\eqref{eq:dual-objective-strictly-convex}, the
latter given in~\eqref{eq:Hessian-superoperator}. We note here that
an expression for its maximum eigenvalue is given by
\begin{equation}
\max_{\omega:\left\Vert \omega\right\Vert _{2}=1}\sum_{x\in\mathcal{X}}\int_{0}^{1}ds\,\Tr\!\left[\omega^{\dag}e^{s\left(\ln\sigma_{x}-A\right)}\omega e^{\left(1-s\right)\left(\ln\sigma_{x}-A\right)}\right]\leq\Tr\!\left[\sum_{x\in\mathcal{X}}e^{\ln\sigma_{x}-A}\right].\label{eq:upper-bound-hessian-eigenval}
\end{equation}
The upper bound in~\eqref{eq:upper-bound-hessian-eigenval} follows
because
\begin{align}
 & \sum_{x\in\mathcal{X}}\int_{0}^{1}ds\,\Tr\!\left[\omega^{\dag}e^{s\left(\ln\sigma_{x}-A\right)}\omega e^{\left(1-s\right)\left(\ln\sigma_{x}-A\right)}\right]\nonumber \\
 & \leq\sum_{x\in\mathcal{X}}\int_{0}^{1}ds\,\left\Vert \omega^{\dag}\right\Vert \left\Vert e^{s\left(\ln\sigma_{x}-A\right)}\right\Vert _{1/s}\left\Vert \omega\right\Vert \left\Vert e^{\left(1-s\right)\left(\ln\sigma_{x}-A\right)}\right\Vert _{1/(1-s)}\\
 & =\sum_{x\in\mathcal{X}}\int_{0}^{1}ds\,\left(\Tr\!\left[e^{\ln\sigma_{x}-A}\right]\right)^{s}\left(\Tr\!\left[e^{\ln\sigma_{x}-A}\right]\right)^{1-s}\\
 & =\sum_{x\in\mathcal{X}}\int_{0}^{1}ds\,\Tr\!\left[e^{\ln\sigma_{x}-A}\right]\\
 & =\sum_{x\in\mathcal{X}}\Tr\!\left[e^{\ln\sigma_{x}-A}\right]\\
 & =\Tr\!\left[\sum_{x\in\mathcal{X}}e^{\ln\sigma_{x}-A}\right],
\end{align}
where we applied the multivariate H\"older inequality for the inequality
(see, e.g., \cite[Eq.~(8)]{Beigi2013}), along with $\left\Vert \omega^{\dag}\right\Vert =\left\Vert \omega\right\Vert \leq\left\Vert \omega\right\Vert _{2}=1$,
and the fact that each $e^{\ln\sigma_{x}-A}$ is positive semidefinite
for the first equality.

At the optimal $A^{\star}$, we can apply the stationarity condition
in~\eqref{eq:stationarity-cond} to conclude that
\begin{equation}
\Tr\!\left[\sum_{x\in\mathcal{X}}e^{\ln\sigma_{x}-A^{\star}}\right]=\Tr[\tau]=1,
\end{equation}
where we used the fact that $\tau$ is a density operator. Thus,
for gradient descent to converge locally near the optimum $A^{\star}$,
it suffices to pick $\eta<2$ (see \cite[Section~2.1.2]{Nesterov2018}
and \cite[Section~3.2]{Bubeck2015}). Indeed, since the Hessian superoperator
depends continuously on $A$, the above estimate implies that for
every $\varepsilon>0$, there exists a neighborhood of $A^{\star}$
such that in which the Hessian norm is bounded from above by $1+\varepsilon$.
Consequently, the gradient is locally Lipschitz, and standard convergence
results for gradient descent on convex functions functions imply local
convergence whenever $\eta\in\left(0,\frac{2}{1+\varepsilon}\right)$.
Since $\varepsilon$ may be chosen arbitrarily small by shrinking
the neighborhood, an arbitrary fixed step size $\eta\in\left(0,2\right)$
yields local convergence.

\section{Proof of Corollary~\ref{cor:classical-case} (classical case)}

\label{app:classical-case}

Let us now consider the commuting case, i.e., when $\left[\tau,\sigma_{x}\right]=\left[\sigma_{x},\sigma_{x'}\right]=0$
for all $x,x'\in\mathcal{X}$. This means that there is an orthonormal
basis $\left\{ |z\rangle\right\} _{z=0}^{d-1}$ such that
\begin{align}
\tau & =\sum_{z=0}^{d-1}\tau_{z}|z\rangle\!\langle z|,\\
\sigma_{x} & =\sum_{z=0}^{d-1}\sigma_{x,z}|z\rangle\!\langle z|.
\end{align}
Then the optimization in~\eqref{eq:dual-bayesian-min-change} reduces
to
\begin{align}
 & \inf_{A\in\Herm}\left\{ \Tr\!\left[A\tau\right]+\sum_{x\in\mathcal{X}}\Tr\!\left[e^{\ln\sigma_{x}-A}\right]\right\} \nonumber \\
 & =\inf_{A\in\Herm}\left\{ \Tr\!\left[A\sum_{z=0}^{d-1}\tau_{z}|z\rangle\!\langle z|\right]+\sum_{x\in\mathcal{X}}\Tr\!\left[\exp\!\left(\ln\!\left(\sum_{z=0}^{d-1}\sigma_{x,z}|z\rangle\!\langle z|\right)-A\right)\right]\right\} \\
 & =\inf_{A\in\Herm}\left\{ \sum_{z=0}^{d-1}\tau_{z}\langle z|A|z\rangle+\sum_{x\in\mathcal{X}}\Tr\!\left[\exp\!\left(\sum_{z=0}^{d-1}\ln\!\left(\sigma_{x,z}\right)|z\rangle\!\langle z|-A\right)\right]\right\} .
\end{align}
Now, defining the unitaries
\begin{equation}
U_{y}\coloneqq\sum_{z=0}^{d-1}e^{2\pi iyz/d}|z\rangle\!\langle z|,
\end{equation}
and observing that
\begin{equation}
\frac{1}{d}\sum_{y=0}^{d-1}U_{y}\omega U_{y}^{\dag}=\sum_{z=0}^{d-1}|z\rangle\!\langle z|\omega|z\rangle\!\langle z|,
\end{equation}
consider that
\begin{align}
& \Tr\!\left[\exp\left(\sum_{z=0}^{d-1}\ln\!\left(\sigma_{x,z}\right)|z\rangle\!\langle z|-A\right)\right] \notag \\
& =\frac{1}{d}\sum_{y=0}^{d-1}\Tr\!\left[U_{y}\exp\!\left(\sum_{z=0}^{d-1}\ln\!\left(\sigma_{x,z}\right)|z\rangle\!\langle z|-A\right)U_{y}^{\dag}\right]\\
 & =\frac{1}{d}\sum_{y=0}^{d-1}\Tr\!\left[\exp\!\left(U_{y}\left[\sum_{z=0}^{d-1}\ln\!\left(\sigma_{x,z}\right)|z\rangle\!\langle z|-A\right]U_{y}^{\dag}\right)\right]\\
 & =\frac{1}{d}\sum_{y=0}^{d-1}\Tr\!\left[\exp\!\left(\sum_{z=0}^{d-1}\ln\!\left(\sigma_{x,z}\right)|z\rangle\!\langle z|-U_{y}AU_{y}^{\dag}\right)\right]\\
 & \geq\Tr\!\left[\exp\!\left(\sum_{z=0}^{d-1}\ln\!\left(\sigma_{x,z}\right)|z\rangle\!\langle z|-\frac{1}{d}\sum_{y=0}^{d-1}U_{y}AU_{y}^{\dag}\right)\right]\\
 & =\Tr\!\left[\exp\!\left(\sum_{z=0}^{d-1}\ln\!\left(\sigma_{x,z}\right)|z\rangle\!\langle z|-\sum_{z=0}^{d-1}|z\rangle\!\langle z|A|z\rangle\!\langle z|\right)\right]\\
 & =\Tr\!\left[\exp\!\left(\sum_{z=0}^{d-1}\left(\ln\!\left(\sigma_{x,z}\right)-\langle z|A|z\rangle\right)|z\rangle\!\langle z|\right)\right]\\
 & =\Tr\!\left[\sum_{z=0}^{d-1}\exp\!\left(\ln\!\left(\sigma_{x,z}\right)-\langle z|A|z\rangle\right)|z\rangle\!\langle z|\right]\\
 & =\sum_{z=0}^{d-1}\exp\!\left(\ln\!\left(\sigma_{x,z}\right)-\langle z|A|z\rangle\right).
\end{align}
For the inequality, we employed the convexity of the function $B\to\Tr\!\left[e^{B}\right]$
\cite[Theorem~2.10]{Carlen2010}. So this implies that
\begin{align}
 & \inf_{A\in\Herm}\left\{ \Tr\!\left[A\tau\right]+\sum_{x\in\mathcal{X}}\Tr\!\left[e^{\ln\sigma_{x}-A}\right]\right\} \nonumber \\
 & =\inf_{A\in\Herm}\left\{ \sum_{z=0}^{d-1}\tau_{z}\langle z|A|z\rangle+\sum_{x\in\mathcal{X}}\sum_{z=0}^{d-1}\exp\!\left(\ln\!\left(\sigma_{x,z}\right)-\langle z|A|z\rangle\right)\right\} \\
 & =\inf_{A\in\Herm}\left\{ \sum_{z=0}^{d-1}\tau_{z}\langle z|A|z\rangle+\sum_{x\in\mathcal{X}}\exp\!\left(\ln\!\left(\sigma_{x,z}\right)-\langle z|A|z\rangle\right)\right\} \\
 & =\inf_{a_{z}\in\mathbb{R}\,\forall z}\left\{ \sum_{z=0}^{d-1}\tau_{z}a_{z}+\sum_{x\in\mathcal{X}}\exp\!\left(\ln\!\left(\sigma_{x,z}\right)-a_{z}\right)\right\} \\
 & =\inf_{a_{z}\in\mathbb{R}\,\forall z}\left\{ \sum_{z=0}^{d-1}\tau_{z}a_{z}+\sum_{x\in\mathcal{X}}\sigma_{x,z}e^{-a_{z}}\right\} \\
 & =\sum_{z=0}^{d-1}\inf_{a_{z}\in\mathbb{R}}\left\{ \tau_{z}a_{z}+\sum_{x\in\mathcal{X}}\sigma_{x,z}e^{-a_{z}}\right\} .
\end{align}
The optimality equation in~\eqref{eq:constraint-optimal-povm-1} then
becomes
\begin{equation}
\sum_{x\in\mathcal{X}}\sum_{z=0}^{d-1}\sigma_{x,z}e^{-a_{z}}|z\rangle\!\langle z|=\sum_{z=0}^{d-1}\tau_{z}|z\rangle\!\langle z|,
\end{equation}
which implies that
\begin{equation}
\sum_{z=0}^{d-1}\left(\sum_{x\in\mathcal{X}}\sigma_{x,z}\right)e^{-a_{z}}|z\rangle\!\langle z|=\sum_{z=0}^{d-1}\tau_{z}|z\rangle\!\langle z|.
\end{equation}
Thus,
\begin{equation}
a_{z}^{\star}=\ln\!\left(\frac{\sum_{x\in\mathcal{X}}\sigma_{x,z}}{\tau_{z}}\right).
\end{equation}
Then the optimal measurement in~\eqref{eq:optimal-measurement} reduces
to
\begin{align}
M_{x}^{\star} & =\tau^{-\frac{1}{2}}\left(e^{\ln\sigma_{x}-A^{\star}}\right)\tau^{-\frac{1}{2}}\\
 & =\tau^{-1}\left(e^{\ln\sigma_{x}-A^{\star}}\right)\\
 & =\sum_{z}\tau_{z}^{-1}\sigma_{x,z}e^{-a_{z}^{\star}}|z\rangle\!\langle z|\\
 & =\sum_{z}\tau_{z}^{-1}\sigma_{x,z}e^{-\ln\frac{\sum_{x\in\mathcal{X}}\sigma_{x,z}}{\tau_{z}}}|z\rangle\!\langle z|\\
 & =\sum_{z}\frac{\sigma_{x,z}}{\sum_{x\in\mathcal{X}}\sigma_{x,z}}|z\rangle\!\langle z|\\
 & =\sigma_{x}\sigma^{-1},
\end{align}
thus establishing the claim in~\eqref{eq:optimal-meas-classical}.
Additionally, defining
\begin{equation}
\sigma_{z}\coloneqq\sum_{x\in\mathcal{X}}\sigma_{x,z},
\end{equation}
the optimal objective function value is given by
\begin{align}
 & 1-\sum_{z=0}^{d-1}\left(\tau_{z}\ln\!\left(\frac{\sum_{x\in\mathcal{X}}\sigma_{x,z}}{\tau_{z}}\right)+\sum_{x\in\mathcal{X}}\sigma_{x,z}e^{-\ln\frac{\sum_{x\in\mathcal{X}}\sigma_{x,z}}{\tau_{z}}}\right)\nonumber \\
 & =1-\sum_{z=0}^{d-1}\left(\tau_{z}\ln\!\left(\frac{\sum_{x\in\mathcal{X}}\sigma_{x,z}}{\tau_{z}}\right)+\sum_{x\in\mathcal{X}}\frac{\tau_{z}\sigma_{x,z}}{\sum_{x\in\mathcal{X}}\sigma_{x,z}}\right)\\
 & =1-\sum_{z=0}^{d-1}\left(\tau_{z}\ln\!\left(\frac{\sigma_{z}}{\tau_{z}}\right)+\tau_{z}\right)\\
 & =\sum_{z=0}^{d-1}\tau_{z}\ln\!\left(\frac{\tau_{z}}{\sigma_{z}}\right)\\
 & =D(\tau\|\sigma),
\end{align}
thus establishing the claim in~\eqref{eq:dual-bayesian-min-change-classical}.

\section{Proof of Corollary~\ref{cor:fermi-dirac} (binary case)}

\label{app:fermi-dirac}

We now consider a special case of Corollary~\ref{cor:softmax} in
which $\left|\mathcal{X}\right|=2$. We derive it from the beginning,
as it seems simpler to do so. Consider that
\begin{align}
\min_{M_{\mathcal{X}}}D(Q_{\mathrm{rev}}(M_{\mathcal{X}})\|Q_{\mathrm{fwd}}) & =\min_{M:0\leq M\leq I}\left\{ D(M/d\|\sigma_{0})+D(\left(I-M\right)/d\|\sigma_{1})\right\} \\
 & =\min_{X:0\leq X\leq I/d}\left\{ D(X\|\sigma_{0})+D(I/d-X\|\sigma_{1})\right\} .
\end{align}
The function
\begin{equation}
M\mapsto D(M/d\|\sigma_{0})+D(\left(I-M\right)/d\|\sigma_{1})
\end{equation}
is strictly convex over the domain $0\leq M\leq I$ because it can
be rewritten as
\begin{multline}
D(M/d\|\sigma_{0})+D(\left(I-M\right)/d\|\sigma_{1})=\Tr\!\left[M/d\ln\!\left(M/d\right)\right]+\Tr\!\left[\left(I-M\right)/d\ln\!\left(\left(I-M\right)/d\right)\right]\\
+\Tr\!\left[M/d\ln\sigma_{0}\right]+\Tr\!\left[\left(I-M\right)/d\ln\sigma_{1}\right].
\end{multline}
The last two terms are affine in $M$ and the first two are strictly
convex in $M$. Thus, the function has a unique global minimum, which
is determined by the first-order stationarity condition. This stationarity
condition is as follows:
\begin{align}
0 & =\frac{\partial}{\partial X}\left(D(X\|\sigma_{0})+D(I/d-X\|\sigma_{1})\right)\nonumber \\
 & =\frac{\partial}{\partial X}\left(\Tr\!\left[X\ln X\right]-\Tr\!\left[X\ln\sigma_{0}\right]+\Tr\!\left[\left(I/d-X\right)\ln\!\left(I/d-X\right)\right]-\Tr\!\left[\left(I/d-X\right)\ln\sigma_{1}\right]\right)\\
 & =\ln X+I-\ln\sigma_{0}-\ln\left(I/d-X\right)-I+\ln\sigma_{1}\\
 & =\ln X-\ln\left(I/d-X\right)+\ln\sigma_{1}-\ln\sigma_{0},
\end{align}
which implies that
\begin{align}
\ln\sigma_{1}-\ln\sigma_{0} & =\ln\!\left(I/d-X\right)-\ln X\\
 & =\ln\!\left(X^{-1}/d-I\right)\\
\implies\quad e^{\ln\sigma_{1}-\ln\sigma_{0}} & =X^{-1}/d-I\\
\implies\quad e^{\ln\sigma_{1}-\ln\sigma_{0}}+I & =X^{-1}/d\\
\implies\quad\left(e^{\ln\sigma_{1}-\ln\sigma_{0}}+I\right)^{-1} & =Xd.
\end{align}
This finally implies that the unique optimal choice of $M$ is
\begin{equation}
M^{\star}=\left(e^{\ln\sigma_{1}-\ln\sigma_{0}}+I\right)^{-1}.
\end{equation}

\section{Proof of Theorem~\ref{thm:SDP-q-softmax}}

\label{app:softmin-SDP}

Consider that
\begin{align}
 & \min_{M_{\mathcal{X}}}\left\{ \sum_{x\in\mathcal{X}}\Tr\!\left[H_{x}M_{x}\right]-TS\!\left(M_{\mathcal{X}}\right):\Tr\!\left[Q_{i}M_{x}\right]=q_{i,x}\forall i,x\right\} \nonumber \\
 & =\min_{M_{\mathcal{X}}}\left\{ \sum_{x\in\mathcal{X}}\Tr\!\left[H_{x}M_{x}\right]-TS\!\left(M_{\mathcal{X}}\right)+\sup_{\mu_{i,x}\in\mathbb{R}\forall i,x}\left\{ \sum_{i,x}\mu_{i,x}\left(q_{i,x}-\Tr\!\left[Q_{i}M_{x}\right]\right)\right\} \right\} \\
 & =\min_{M_{\mathcal{X}}}\sup_{\mu_{i,x}\in\mathbb{R}\forall i,x}\left\{ \sum_{x\in\mathcal{X}}\Tr\!\left[H_{x}M_{x}\right]-TS\!\left(M_{\mathcal{X}}\right)+\sum_{i,x}\mu_{i,x}\left(q_{i,x}-\Tr\!\left[Q_{i}M_{x}\right]\right)\right\} \\
 & =\min_{M_{\mathcal{X}}}\sup_{\mu_{i,x}\in\mathbb{R}\forall i,x}\left\{ \sum_{i,x}\mu_{i,x}q_{i,x}-TS\!\left(M_{\mathcal{X}}\right)+\sum_{x\in\mathcal{X}}\Tr\!\left[M_{x}\left(H_{x}-\sum_{i}\mu_{i,x}Q_{i}\right)\right]\right\} \\
 & =\sup_{\mu_{i,x}\in\mathbb{R}\forall i,x}\min_{M_{\mathcal{X}}}\left\{ \sum_{i,x}\mu_{i,x}q_{i,x}-TS\!\left(M_{\mathcal{X}}\right)+\sum_{x\in\mathcal{X}}\Tr\!\left[M_{x}\left(H_{x}-\sum_{i}\mu_{i,x}Q_{i}\right)\right]\right\} \\
 & =\sup_{\mu_{i,x}\in\mathbb{R}\forall i,x}\left\{ \sum_{i,x}\mu_{i,x}q_{i,x}+\min_{M_{\mathcal{X}}}\left\{ -TS\!\left(M_{\mathcal{X}}\right)+\sum_{x\in\mathcal{X}}\Tr\!\left[M_{x}\left(H_{x}-\sum_{i}\mu_{i,x}Q_{i}\right)\right]\right\} \right\} .
\end{align}
The penultimate equality follows from the Sion minimax theorem, given
that the objective function is convex in $M_{\mathcal{X}}$, the set
of POVMs is compact, and the objective function is linear in $\mu_{i,x}$.
Defining
\begin{equation}
\mu_{x}\cdot Q\equiv\sum_{i}\mu_{i,x}Q_{i},
\end{equation}
consider that
\begin{align}
 & \min_{M_{\mathcal{X}}}\left\{ -TS\!\left(M_{\mathcal{X}}\right)+\sum_{x\in\mathcal{X}}\Tr\!\left[M_{x}\left(H_{x}-\mu_{x}\cdot Q\right)\right]\right\} \nonumber \\
 & =\min_{M_{\mathcal{X}}}\left\{ -TS\!\left(M_{\mathcal{X}}\right)+Td+\sum_{x\in\mathcal{X}}\Tr\!\left[M_{x}\left(H_{x}-\mu_{x}\cdot Q\right)\right]-T\Tr\!\left[M_{x}\right]\right\} \\
 & =Td+\min_{M_{\mathcal{X}}}\left\{ -TS\!\left(M_{\mathcal{X}}\right)+\sum_{x\in\mathcal{X}}\Tr\!\left[M_{x}\left(H_{x}-\mu_{x}\cdot Q\right)\right]-T\Tr\!\left[M_{x}\right]\right\} \\
 & =Td+\notag \\
 & \min_{M_{x}\geq0\forall x}\left\{ -TS\!\left(M_{\mathcal{X}}\right)+\sum_{x\in\mathcal{X}}\Tr\!\left[M_{x}\left(H_{x}-\mu_{x}\cdot Q\right)\right]-T\Tr\!\left[M_{x}\right]+\sup_{A\in\Herm}\Tr\!\left[A\left(\sum_{x\in\mathcal{X}}M_{x}-I\right)\right]\right\} .
\end{align}
Now consider that
\begin{align}
 & \min_{M_{x}\geq0\forall x}\left\{ -TS\!\left(M_{\mathcal{X}}\right)+\sum_{x\in\mathcal{X}}\Tr\!\left[M_{x}\left(H_{x}-\mu_{x}\cdot Q\right)\right]-T\Tr\!\left[M_{x}\right]+\sup_{A\in\Herm}\Tr\!\left[A\left(\sum_{x\in\mathcal{X}}M_{x}-I\right)\right]\right\} \nonumber \\
 & =\min_{M_{x}\geq0\forall x}\sup_{A\in\Herm}\left\{ -TS\!\left(M_{\mathcal{X}}\right)+\sum_{x\in\mathcal{X}}\Tr\!\left[M_{x}\left(H_{x}-\mu_{x}\cdot Q\right)\right]-T\Tr\!\left[M_{x}\right]+\Tr\!\left[A\left(\sum_{x\in\mathcal{X}}M_{x}-I\right)\right]\right\} \\
 & =\min_{M_{x}\geq0\forall x}\sup_{A\in\Herm}\left\{ -\Tr[A]-TS\!\left(M_{\mathcal{X}}\right)+\sum_{x\in\mathcal{X}}\Tr\!\left[M_{x}\left(H_{x}+A-\mu_{x}\cdot Q\right)\right]-T\Tr\!\left[M_{x}\right]\right\} \\
 & =\sup_{A\in\Herm}\min_{M_{x}\geq0\forall x}\left\{ -\Tr[A]-TS\!\left(M_{\mathcal{X}}\right)+\sum_{x\in\mathcal{X}}\Tr\!\left[M_{x}\left(H_{x}+A-\mu_{x}\cdot Q\right)\right]-T\Tr\!\left[M_{x}\right]\right\} \\
 & =\sup_{A\in\Herm}\left\{ -\Tr[A]+\min_{M_{x}\geq0\forall x}\left\{ -TS\!\left(M_{\mathcal{X}}\right)+\sum_{x\in\mathcal{X}}\Tr\!\left[M_{x}\left(H_{x}+A-\mu_{x}\cdot Q\right)\right]-T\Tr\!\left[M_{x}\right]\right\} \right\} .
\end{align}
For the penultimate equality, the minimax equality follows because
the objective function is convex in $M_{x}$ and linear in $A$. It
is also a consequence of Slater's theorem, which holds because the
primal feasible set contains interior points such as $M_{x}=I/\left|\mathcal{X}\right|$.
Now consider that
\begin{align}
 & -TS\!\left(M_{\mathcal{X}}\right)+\sum_{x\in\mathcal{X}}\Tr\!\left[M_{x}\left(H_{x}+A-\mu_{x}\cdot Q\right)\right]-T\Tr\!\left[M_{x}\right]\nonumber \\
 & =T\left(-S\!\left(M_{\mathcal{X}}\right)-\sum_{x\in\mathcal{X}}\Tr\!\left[M_{x}\left(-\frac{1}{T}\left(H_{x}+A-\mu_{x}\cdot Q\right)\right)\right]-\Tr\!\left[M_{x}\right]\right)\\
 & =T\left(-S\!\left(M_{\mathcal{X}}\right)-\sum_{x\in\mathcal{X}}\Tr\!\left[M_{x}\ln e^{-\frac{1}{T}\left(H_{x}+A-\mu_{x}\cdot Q\right)}\right]-\Tr\!\left[M_{x}\right]\right)\\
 & =T\left(\sum_{x\in\mathcal{X}}D\!\left(M_{x}\middle\|e^{-\frac{1}{T}\left(H_{x}+A-\mu_{x}\cdot Q\right)}\right)-\Tr\!\left[M_{x}\right]\right)\\
 & =T\left(\sum_{x\in\mathcal{X}}\widetilde{D}\!\left(M_{x}\middle\|e^{-\frac{1}{T}\left(H_{x}+A-\mu_{x}\cdot Q\right)}\right)-\Tr\!\left[e^{-\frac{1}{T}\left(H_{x}+A-\sum_{i}\mu_{i,x}Q_{i}\right)}\right]\right)\\
 & =T\left(\sum_{x\in\mathcal{X}}\widetilde{D}\!\left(M_{x}\middle\|e^{-\frac{1}{T}\left(H_{x}+A-\sum_{i}\mu_{i,x}Q_{i}\right)}\right)-\Tr\!\left[e^{-\frac{1}{T}\left(H_{x}+A-\sum_{i}\mu_{i,x}Q_{i}\right)}\right]\right).
\end{align}
Then we find that
\begin{align}
 & \sup_{A\in\Herm}\left\{ -\Tr[A]+\min_{M_{x}\geq0\forall x}\left\{ -TS\!\left(M_{\mathcal{X}}\right)+\sum_{x\in\mathcal{X}}\Tr\!\left[M_{x}\left(H_{x}+A-\mu_{x}\cdot Q\right)\right]-T\Tr\!\left[M_{x}\right]\right\} \right\} \nonumber \\
 & =\sup_{A\in\Herm}\left\{ -\Tr[A]+\min_{M_{x}\geq0\forall x}\left\{ T\left(\sum_{x\in\mathcal{X}}\widetilde{D}\!\left(M_{x}\middle\|e^{-\frac{1}{T}\left(H_{x}+A-\mu_{x}\cdot Q\right)}\right)-\Tr\!\left[e^{-\frac{1}{T}\left(H_{x}+A-\mu_{x}\cdot Q\right)}\right]\right)\right\} \right\} \label{eq:gen-rel-ent-to-min}\\
 & =\sup_{A\in\Herm}\left\{ -\Tr[A]-T\sum_{x\in\mathcal{X}}\Tr\!\left[e^{-\frac{1}{T}\left(H_{x}+A-\mu_{x}\cdot Q\right)}\right]\right\} \\
 & =-\inf_{A\in\Herm}\left\{ \Tr[A]+T\sum_{x\in\mathcal{X}}\Tr\!\left[e^{-\frac{1}{T}\left(H_{x}+A-\mu_{x}\cdot Q\right)}\right]\right\} ,\label{eq:inf-softmax}
\end{align}
so that
\begin{align}
 & \min_{M_{\mathcal{X}}}\left\{ \sum_{x\in\mathcal{X}}\Tr\!\left[H_{x}M_{x}\right]-TS\!\left(M_{\mathcal{X}}\right):\Tr\!\left[Q_{i}M_{x}\right]=q_{i,x}\forall i,x\right\} \nonumber \\
 & =\sup_{\mu_{i,x}\in\mathbb{R}\forall i,x}\left\{ \sum_{i,x}\mu_{i,x}q_{i,x}+Td-\inf_{A\in\Herm}\left\{ \Tr[A]+T\sum_{x\in\mathcal{X}}\Tr\!\left[e^{-\frac{1}{T}\left(H_{x}+A-\mu_{x}\cdot Q\right)}\right]\right\} \right\} \\
 & =Td+\sup_{\mu_{i,x}\in\mathbb{R}\forall i,x}\left\{ \sum_{i,x}\mu_{i,x}q_{i,x}-\inf_{A\in\Herm}\left\{ \Tr[A]+T\sum_{x\in\mathcal{X}}\Tr\!\left[e^{-\frac{1}{T}\left(H_{x}+A-\mu_{x}\cdot Q\right)}\right]\right\} \right\} .
\end{align}
A similar analysis, as in~\eqref{eq:stationarity-cond}, of the stationarity
condition for~\eqref{eq:inf-softmax} allows us to conclude that an
optimal $A^{\star}$ satisfies
\begin{equation}
\sum_{x\in\mathcal{X}}e^{-\frac{1}{T}\left(H_{x}+A^{\star}-\mu_{x}\cdot Q\right)}=I.
\end{equation}
Furthermore, based on the faithfulness of the generalized relative
entropy in~\eqref{eq:gen-rel-ent-to-min}, an optimal measurement
operator $M_{x}^{\star}$ is as follows:
\begin{equation}
M_{x}^{\star}=e^{-\frac{1}{T}\left(H_{x}+A^{\star}-\mu_{x}\cdot Q\right)}.
\end{equation}

\section{Proof of Theorem~\ref{thm:additivity}}

\label{app:additivity}

By choosing the measurement $M_{\mathcal{X}_{1}\times\mathcal{X}_{2}}$
to be a product measurement with measurement operators of the form
$M_{x_{1}}^{(1)}\otimes M_{x_{2}}^{(2)}$, it follows that $Q_{\mathrm{rev}}(\tau_{1}\otimes\tau_{2},M_{\mathcal{X}_{1}\times\mathcal{X}_{2}})$
is a tensor product of $Q_{\mathrm{rev}}(\tau_{1},M_{\mathcal{X}_{1}})$
and $Q_{\mathrm{rev}}(\tau_{2},M_{\mathcal{X}_{2}})$:
\begin{equation}
Q_{\mathrm{rev}}(\tau_{1}\otimes\tau_{2},M_{\mathcal{X}_{1}\times\mathcal{X}_{2}})\simeq Q_{\mathrm{rev}}(\tau_{1},M_{\mathcal{X}_{1}})\otimes Q_{\mathrm{rev}}(\tau_{2},M_{\mathcal{X}_{2}}),
\end{equation}
where the $\simeq$ symbol indicates equality up to permutation of
systems. Given that $Q_{\mathrm{fwd}}(p_{X_{1}}\otimes p_{X_{2}},\mathcal{E}_{1}\otimes\mathcal{E}_{2})$
is a tensor product of $Q_{\mathrm{fwd}}(p_{X_{1}},\mathcal{E}_{1})$
and $Q_{\mathrm{fwd}}(p_{X_{2}},\mathcal{E}_{2})$, i.e.,
\begin{equation}
Q_{\mathrm{fwd}}(p_{X_{1}}\otimes p_{X_{2}},\mathcal{E}_{1}\otimes\mathcal{E}_{2})\simeq Q_{\mathrm{fwd}}(p_{X_{1}},\mathcal{E}_{1})\otimes Q_{\mathrm{fwd}}(p_{X_{2}},\mathcal{E}_{2}),
\end{equation}
we apply additivity of quantum relative entropy for tensor products
to conclude that
\begin{equation}
D(Q_{\mathrm{rev}}(\tau_{1}\otimes\tau_{2},M_{\mathcal{X}_{1}\times\mathcal{X}_{2}})\|Q_{\mathrm{fwd}}(p_{X_{1}}\otimes p_{X_{2}},\mathcal{E}_{1}\otimes\mathcal{E}_{2}))=\sum_{i=1}^{2}D(Q_{\mathrm{rev}}(\tau_{i},M_{\mathcal{X}_{i}})\|Q_{\mathrm{fwd}}(p_{X_{i}},\mathcal{E}_{i}))
\end{equation}
in this case. Since the choice of $M_{\mathcal{X}_{1}}$ and $M_{\mathcal{X}_{2}}$
is arbitrary, we conclude that
\begin{multline}
\min_{M_{\mathcal{X}_{1}\times\mathcal{X}_{2}}}D(Q_{\mathrm{rev}}(\tau_{1}\otimes\tau_{2},M_{\mathcal{X}_{1}\times\mathcal{X}_{2}})\|Q_{\mathrm{fwd}}(p_{X_{1}}\otimes p_{X_{2}},\mathcal{E}_{1}\otimes\mathcal{E}_{2}))\\
\leq\sum_{i=1}^{2}\min_{M_{\mathcal{X}_{i}}}D(Q_{\mathrm{rev}}(\tau_{i},M_{\mathcal{X}_{i}})\|Q_{\mathrm{fwd}}(p_{X_{i}},\mathcal{E}_{i})),\label{eq:subadd-proof}
\end{multline}
as the optimization on the left-hand side is over all possible joint
measurements.

To see that the inequality holds in the opposite direction, let us
recall that the quantum relative entropy is superadditive in the following
sense (see, e.g., \cite[Proposition~2]{Capel2018Superadditivity}):
\begin{equation}
D(\xi_{AB}\|\gamma_{A}\otimes\zeta_{B})\geq D(\xi_{A}\|\gamma_{A})+D(\xi_{B}\|\zeta_{B}),\label{eq:superadd-rel-ent}
\end{equation}
where $\xi_{AB}$, $\gamma_{A}$, and $\zeta_{B}$ are states, $\xi_{A}\coloneqq\Tr_{B}[\xi_{AB}]$,
and $\xi_{B}\coloneqq\Tr_{A}[\xi_{AB}]$. Now picking a general measurement
$M_{\mathcal{X}_{1}\times\mathcal{X}_{2}}$, we apply~\eqref{eq:superadd-rel-ent}
to conclude that
\begin{align}
 & D(Q_{\mathrm{rev}}(\tau_{1}\otimes\tau_{2},M_{\mathcal{X}_{1}\times\mathcal{X}_{2}})\|Q_{\mathrm{fwd}}(p_{X_{1}}\otimes p_{X_{2}},\mathcal{E}_{1}\otimes\mathcal{E}_{2}))\nonumber \\
 & =D(Q_{\mathrm{rev}}(\tau_{1}\otimes\tau_{2},M_{\mathcal{X}_{1}\times\mathcal{X}_{2}})\|Q_{\mathrm{fwd}}(p_{X_{1}},\mathcal{E}_{1})\otimes Q_{\mathrm{fwd}}(p_{X_{2}},\mathcal{E}_{2}))\label{eq:superadd-1}\\
 & \geq D(\Tr_{2}\!\left[Q_{\mathrm{rev}}(\tau_{1}\otimes\tau_{2},M_{\mathcal{X}_{1}\times\mathcal{X}_{2}})\right]\|Q_{\mathrm{fwd}}(p_{X_{1}},\mathcal{E}_{1}))+D(\Tr_{1}\!\left[Q_{\mathrm{rev}}(\tau_{1}\otimes\tau_{2},M_{\mathcal{X}_{1}\times\mathcal{X}_{2}})\right]\|Q_{\mathrm{fwd}}(p_{X_{2}},\mathcal{E}_{2}))\\
 & =D(Q_{\mathrm{rev}}(\tau_{1},N_{\mathcal{X}_{1}}^{\left(1\right)})\|Q_{\mathrm{fwd}}(p_{X_{1}},\mathcal{E}_{1}))+D(Q_{\mathrm{rev}}(\tau_{2},N_{\mathcal{X}_{2}}^{\left(2\right)})\|Q_{\mathrm{fwd}}(p_{X_{2}},\mathcal{E}_{2}))\\
 & \geq\sum_{i=1}^{2}\min_{M_{\mathcal{X}_{i}}}D(Q_{\mathrm{rev}}(\tau_{i},M_{\mathcal{X}_{i}})\|Q_{\mathrm{fwd}}(p_{X_{i}},\mathcal{E}_{i})).\label{eq:superadd-last}
\end{align}
The last equality follows because
\begin{align}
\Tr_{2}\!\left[Q_{\mathrm{rev}}(\tau_{1}\otimes\tau_{2},M_{\mathcal{X}_{1}\times\mathcal{X}_{2}})\right] & =\Tr_{2}\!\left[\sum_{\substack{x_{1}\in\mathcal{X}_{1},\\
x_{2}\in\mathcal{X}_{2}
}
}\left(\tau_{1}\otimes\tau_{2}\right)^{\frac{1}{2}}M_{x_{1}x_{2}}\left(\tau_{1}\otimes\tau_{2}\right)^{\frac{1}{2}}\otimes|x_{1}\rangle\!\langle x_{1}|\otimes|x_{2}\rangle\!\langle x_{2}|\right]\\
 & =\sum_{\substack{x_{1}\in\mathcal{X}_{1},\\
x_{2}\in\mathcal{X}_{2}
}
}\Tr_{2}\left[\left(\tau_{1}\otimes\tau_{2}\right)^{\frac{1}{2}}M_{x_{1}x_{2}}\left(\tau_{1}\otimes\tau_{2}\right)^{\frac{1}{2}}\right]\otimes|x_{1}\rangle\!\langle x_{1}|\\
 & =\sum_{x_{1}\in\mathcal{X}_{1}}\tau_{1}^{\frac{1}{2}}\Tr_{2}\left[\left(I\otimes\tau_{2}^{\frac{1}{2}}\right)\sum_{x_{2}\in\mathcal{X}_{2}}M_{x_{1}x_{2}}\left(I\otimes\tau_{2}^{\frac{1}{2}}\right)\right]\tau_{1}^{\frac{1}{2}}\otimes|x_{1}\rangle\!\langle x_{1}|\\
 & =\sum_{x_{1}\in\mathcal{X}_{1}}\tau_{1}^{\frac{1}{2}}N_{x_{1}}^{(1)}\tau_{1}^{\frac{1}{2}}\otimes|x_{1}\rangle\!\langle x_{1}|,
\end{align}
where
\begin{equation}
N_{x_{1}}^{(1)}\coloneqq\Tr_{2}\left[\left(I\otimes\tau_{2}\right)\sum_{x_{2}\in\mathcal{X}_{2}}M_{x_{1}x_{2}}\right].
\end{equation}
Thus, $\left(N_{x_{1}}^{(1)}\right)_{x_{1}\in\mathcal{X}_{1}}$ is
a particular measurement because $N_{x_{1}}^{(1)}\geq0$ for all $x_{1}\in\mathcal{X}_{1}$
and
\begin{align}
\sum_{x_{1}\in\mathcal{X}_{1}}N_{x_{1}}^{(1)} & =\sum_{x_{1}\in\mathcal{X}_{1}}\Tr_{2}\left[\left(I\otimes\tau_{2}\right)\sum_{x_{2}\in\mathcal{X}_{2}}M_{x_{1}x_{2}}\right]\\
 & =\Tr_{2}\left[\left(I\otimes\tau_{2}\right)\sum_{x_{1}\in\mathcal{X}_{1}}\sum_{x_{2}\in\mathcal{X}_{2}}M_{x_{1}x_{2}}\right]\\
 & =\Tr_{2}\left[\left(I\otimes\tau_{2}\right)\right]\\
 & =I.
\end{align}
Similarly, 
\begin{align}
\Tr_{1}\!\left[Q_{\mathrm{rev}}(\tau_{1}\otimes\tau_{2},M_{\mathcal{X}_{1}\times\mathcal{X}_{2}})\right] & =\Tr_{1}\!\left[\sum_{\substack{x_{1}\in\mathcal{X}_{1},\\
x_{2}\in\mathcal{X}_{2}
}
}\left(\tau_{1}\otimes\tau_{2}\right)^{\frac{1}{2}}M_{x_{1}x_{2}}\left(\tau_{1}\otimes\tau_{2}\right)^{\frac{1}{2}}\otimes|x_{1}\rangle\!\langle x_{1}|\otimes|x_{2}\rangle\!\langle x_{2}|\right]\\
 & =\sum_{\substack{x_{1}\in\mathcal{X}_{1},\\
x_{2}\in\mathcal{X}_{2}
}
}\Tr_{1}\!\left[\left(\tau_{1}\otimes\tau_{2}\right)^{\frac{1}{2}}M_{x_{1}x_{2}}\left(\tau_{1}\otimes\tau_{2}\right)^{\frac{1}{2}}\right]\otimes|x_{2}\rangle\!\langle x_{2}|\\
 & =\sum_{x_{2}\in\mathcal{X}_{2}}\tau_{2}^{\frac{1}{2}}\Tr_{1}\!\left[\left(\tau_{1}\otimes I\right)^{\frac{1}{2}}\sum_{x_{1}\in\mathcal{X}_{1}}M_{x_{1}x_{2}}\left(\tau_{1}\otimes I\right)^{\frac{1}{2}}\right]\tau_{2}^{\frac{1}{2}}\otimes|x_{2}\rangle\!\langle x_{2}|\\
 & =\sum_{x_{2}\in\mathcal{X}_{2}}\tau_{2}^{\frac{1}{2}}N_{x_{2}}^{(2)}\tau_{2}^{\frac{1}{2}}\otimes|x_{2}\rangle\!\langle x_{2}|,
\end{align}
where
\begin{equation}
N_{x_{2}}^{(2)}\coloneqq\Tr_{1}\!\left[\left(\tau_{1}\otimes I\right)\sum_{x_{1}\in\mathcal{X}_{1}}M_{x_{1}x_{2}}\right].
\end{equation}
Thus, $\left(N_{x_{2}}^{(2)}\right)_{x_{2}\in\mathcal{X}_{2}}$ is
a particular measurement because $N_{x_{2}}^{(2)}\geq0$ for all $x_{2}\in\mathcal{X}_{2}$
and
\begin{align}
\sum_{x_{2}\in\mathcal{X}_{2}}N_{x_{2}}^{(2)} & =\sum_{x_{2}\in\mathcal{X}_{2}}\Tr_{1}\!\left[\left(\tau_{1}\otimes I\right)\sum_{x_{1}\in\mathcal{X}_{1}}M_{x_{1}x_{2}}\right]\\
 & =\Tr_{1}\!\left[\left(\tau_{1}\otimes I\right)\sum_{x_{1}\in\mathcal{X}_{1}}\sum_{x_{2}\in\mathcal{X}_{2}}M_{x_{1}x_{2}}\right]\\
 & =\Tr_{1}\left[\left(\tau_{1}\otimes I\right)\right]\\
 & =I.
\end{align}
Since the inequalities in~\eqref{eq:superadd-1}--\eqref{eq:superadd-last}
hold for every possible joint measurement, we conclude that
\begin{multline}
\min_{M_{\mathcal{X}_{1}\times\mathcal{X}_{2}}}D(Q_{\mathrm{rev}}(\tau_{1}\otimes\tau_{2},M_{\mathcal{X}_{1}\times\mathcal{X}_{2}})\|Q_{\mathrm{fwd}}(p_{X_{1}}\otimes p_{X_{2}},\mathcal{E}_{1}\otimes\mathcal{E}_{2}))\\
\geq\sum_{i=1}^{2}\min_{M_{\mathcal{X}_{i}}}D(Q_{\mathrm{rev}}(\tau_{i},M_{\mathcal{X}_{i}})\|Q_{\mathrm{fwd}}(p_{X_{i}},\mathcal{E}_{i})).\label{eq:superadd-proof}
\end{multline}
Combining~\eqref{eq:subadd-proof} and~\eqref{eq:superadd-proof}
concludes the proof.

\section{Proofs for Section~\ref{subsec:Fermi=002013Dirac-thermal-measurements-QHT}}

\label{app:Proofs-FD-QHT}

In this appendix, we provide several lemmas and proofs needed for
or claimed in Section~\ref{subsec:Fermi=002013Dirac-thermal-measurements-QHT}.

\subsection{Proof of Equation~\eqref{eq:FD-err-prob-tanh}}

Consider that
\begin{align}
\left(e^{-A}+I\right)^{-1} & =\frac{1}{2}\left(I+\tanh\!\left(A/2\right)\right),\\
\left(e^{A}+I\right)^{-1} & =I-\left(e^{-A}+I\right)^{-1}\\
 & =I-\frac{1}{2}\left(I+\tanh\!\left(A/2\right)\right)\\
 & =\frac{1}{2}\left(I-\tanh\!\left(A/2\right)\right),
\end{align}
implying that
\begin{align}
 & \Tr\!\left[\left(e^{-A}+I\right)^{-1}\sigma_{0}\right]+\Tr\!\left[\left(e^{A}+I\right)^{-1}\sigma_{1}\right]\nonumber \\
 & =\Tr\!\left[\frac{1}{2}\left(I+\tanh\!\left(A/2\right)\right)\sigma_{0}\right]+\Tr\!\left[\frac{1}{2}\left(I-\tanh\!\left(A/2\right)\right)\sigma_{1}\right]\\
 & =\frac{1}{2}-\frac{1}{2}\Tr\!\left[\tanh\!\left(A/2\right)\left(\sigma_{1}-\sigma_{0}\right)\right]\\
 & =\frac{1}{2}-\frac{1}{2}\Tr\!\left[\tanh\!\left(A/2\right)\Delta\right],
\end{align}
thus concluding the proof of Eq.~\eqref{eq:FD-err-prob-tanh}.

\subsection{Supplementary lemma}
\begin{lem}
\label{lem:FD-scalar-bnd-s}The following inequality holds for all
$s\in\left[0,1\right]$:
\begin{equation}
(1+e^{x})^{-1}\leq e^{-h(s)}e^{-sx},
\end{equation}
where the binary entropy is defined as $h(s)\coloneqq-s\ln s-\left(1-s\right)\ln\!\left(1-s\right)$.
\end{lem}

\begin{proof}
We would like the following inequality to hold for all $x\in\mathbb{R}$,
for some function $f(s)$:
\begin{equation}
(1+e^{x})^{-1}\leq f(s)e^{-sx}.
\end{equation}
This is equivalent to
\begin{equation}
e^{sx}(1+e^{x})^{-1}\leq f(s).
\end{equation}
Thus, we can find $f(s)$ if we can evaluate the following:
\begin{equation}
\sup_{x\in\mathbb{R}}g(x,s),
\end{equation}
where
\begin{equation}
g(x,s)\coloneqq e^{sx}(1+e^{x})^{-1}.
\end{equation}
Then consider that
\begin{align}
\frac{\partial}{\partial x}g(x,s) & =se^{sx}(1+e^{x})^{-1}-e^{sx}(1+e^{x})^{-2}e^{x}\\
 & =e^{sx}(1+e^{x})^{-2}\left(s(1+e^{x})-e^{x}\right).
\end{align}
Now we set this derivative equal to zero and solve for $x$:
\begin{align}
0 & =\frac{\partial}{\partial x}g(x,s)\\
\Longleftrightarrow\qquad0 & =e^{sx}(1+e^{x})^{-2}\left(s(1+e^{x})-e^{x}\right)\\
\Longleftrightarrow\qquad0 & =s(1+e^{x})-e^{x}\\
\Longleftrightarrow\qquad\frac{e^{x}}{1+e^{x}} & =s\\
\Longleftrightarrow\qquad\frac{1}{e^{-x}+1} & =s\\
\Longleftrightarrow\qquad\frac{1}{s} & =e^{-x}+1\\
\Longleftrightarrow\qquad\frac{1}{s}-1 & =e^{-x}\\
\Longleftrightarrow\qquad\ln\!\left(\frac{1-s}{s}\right) & =-x\\
\Longleftrightarrow\qquad\ln\!\left(\frac{s}{1-s}\right) & =x.
\end{align}
We have found the unique stationary point. Now observe that, for $s\in\left(0,1\right)$,
\begin{align}
\lim_{x\to-\infty}g(x,s) & =\lim_{x\to-\infty}\frac{e^{sx}}{1+e^{x}}=0,\\
\lim_{x\to\infty}g(x,s) & =\lim_{x\to\infty}\frac{e^{sx}}{1+e^{x}}=\lim_{x\to\infty}e^{-\left(1-s\right)x}=0.
\end{align}
Since $g(x,s)$ is continuous on $\mathbb{R}$, tends to zero at both
ends of the real line, and has a unique stationary point, this stationary
point must be the global maximizer.

Plugging this value of $x$ into $g(x,s)$, we find that
\begin{align}
e^{sx}(1+e^{x})^{-1} & =e^{s\ln\left(\frac{s}{1-s}\right)}\left(1+e^{\ln\left(\frac{s}{1-s}\right)}\right)^{-1}\\
 & =\left(\frac{s}{1-s}\right)^{s}\left(1+\frac{s}{1-s}\right)^{-1}\\
 & =\left(\frac{s}{1-s}\right)^{s}\left(1-s\right)\\
 & =s^{s}\left(1-s\right)^{1-s}\\
 & =e^{-h(s)}.
\end{align}
This concludes the proof.
\end{proof}

\subsection{Proof of Lemma~\ref{lem:usual-chernoff-to-modified}}

\label{app:usual-chernoff-to-modified}

Consider that, for $\alpha\in\left(0,1\right)$, by applying \cite[Lemma~3]{Berta2017},
\begin{align}
\Tr\!\left[\rho_{0}^{\alpha}\rho_{1}^{1-\alpha}\right] & \leq\inf_{\omega>0}\left(\Tr\!\left[\omega\rho_{0}\right]\right)^{\alpha}\left(\Tr\!\left[\omega^{\frac{\alpha}{\alpha-1}}\rho_{1}\right]\right)^{1-\alpha}\label{eq:petz-renyi-measured}\\
 & \leq\left(\Tr\!\left[e^{\left(1-s\right)\left(\ln\rho_{1}-\ln\rho_{0}\right)}\rho_{0}\right]\right)^{\alpha}\left(\Tr\!\left[\left(e^{\left(1-s\right)\left(\ln\rho_{1}-\ln\rho_{0}\right)}\right)^{\frac{\alpha}{\alpha-1}}\rho_{1}\right]\right)^{1-\alpha},
\end{align}
where the last inequality follows by picking $\omega=e^{\left(1-s\right)\left(\ln\rho_{1}-\ln\rho_{0}\right)}$.
Now pick $\alpha=s$. This implies that
\begin{align}
\Tr\!\left[\rho_{0}^{s}\rho_{1}^{1-s}\right] & \leq\left(\Tr\!\left[e^{\left(1-s\right)\left(\ln\rho_{1}-\ln\rho_{0}\right)}\rho_{0}\right]\right)^{s}\left(\Tr\!\left[\left(e^{\left(1-s\right)\left(\ln\rho_{1}-\ln\rho_{0}\right)}\right)^{\frac{s}{s-1}}\rho_{1}\right]\right)^{1-s}\\
 & =\left(\Tr\!\left[e^{\left(1-s\right)\left(\ln\rho_{1}-\ln\rho_{0}\right)}\rho_{0}\right]\right)^{s}\left(\Tr\!\left[e^{s\left(\ln\rho_{0}-\ln\rho_{1}\right)}\rho_{1}\right]\right)^{1-s}.
\end{align}
We can rewrite this as
\begin{align}
C_{s}(\rho_{0}\|\rho_{1}) & =-\ln\Tr\!\left[\rho_{0}^{s}\rho_{1}^{1-s}\right]\\
 & \geq sC_{s}^{\natural}(\rho_{0}\|\rho_{1})+\left(1-s\right)C_{1-s}^{\natural}(\rho_{1}\|\rho_{0})\\
 & \geq\min\left\{ C_{s}^{\natural}(\rho_{0}\|\rho_{1}),C_{1-s}^{\natural}(\rho_{1}\|\rho_{0})\right\} ,
\end{align}
thus establishing~\eqref{eq:usual-chernoff-to-modified}.

The statement about strict inequality follows because the inequality
in~\eqref{eq:petz-renyi-measured} is strict if the states are positive
definite and do not commute. This follows as a direct consequence
of Lemma~\ref{lem:strict-ineq-petz-measured} below and the variational
characterization of the measured R\'enyi relative entropy from \cite[Lemma~3 \& Theorem~4]{Berta2017}.
To prepare for this, let us recall the Petz--R\'enyi relative entropy
\cite{Petz1985,Petz1986a}, the sandwiched R\'enyi relative entropy
\cite{MuellerLennert2013,Wilde2014}, and the measured R\'enyi relative
entropy~\cite{Fuchs1996}:
\begin{align}
D_{\alpha}(\rho\|\sigma) & \coloneqq\frac{1}{\alpha-1}\ln\Tr\!\left[\rho^{\alpha}\sigma^{1-\alpha}\right],\\
\widetilde{D}_{\alpha}(\rho\|\sigma) & \coloneqq\frac{1}{\alpha-1}\ln\Tr\!\left[\left(\sigma^{\frac{1-\alpha}{2\alpha}}\rho\sigma^{\frac{1-\alpha}{2\alpha}}\right)^{\alpha}\right],\\
D_{\alpha}^{M}(\rho\|\sigma) & \coloneqq\sup_{\mathcal{X},M_{\mathcal{X}}}\frac{1}{\alpha-1}\ln\sum_{x\in\mathcal{X}}\left(\Tr\!\left[M_{x}\rho\right]\right)^{\alpha}\left(\Tr\!\left[M_{x}\sigma\right]\right)^{1-\alpha}.
\end{align}
For $\alpha\in\left[1/2,1\right)$, these are related as follows:
\begin{equation}
D_{\alpha}^{M}(\rho\|\sigma)\leq\widetilde{D}_{\alpha}(\rho\|\sigma)\leq D_{\alpha}(\rho\|\sigma),\label{eq:meas-sand-petz}
\end{equation}
where the first inequality follows from data-processing~\cite{Frank2013}
and the second from \cite[Lemma~3]{Datta2014Limit}.
\begin{rem}
The statement about strict inequality in the following lemma is a
consequence of \cite[Theorem~4.18]{Hiai2017Different} or \cite[Remark~III.12]{Hiai2023TestMeasured}.
We state it here and provide a proof for convenience.
\end{rem}

\begin{lem}
\label{lem:strict-ineq-petz-measured}For all $\alpha\in\left(0,1\right)$,
the following inequality holds
\begin{equation}
D_{\alpha}(\rho\|\sigma)\geq D_{\alpha}^{M}(\rho\|\sigma).
\end{equation}
For positive definite $\rho$ and $\sigma$, the inequality is saturated
if and only if $\left[\rho,\sigma\right]=0$.
\end{lem}

\begin{proof}
The inequality is a direct consequence of the data-processing inequality
for the Petz--R\'enyi relative entropy~\cite{Petz1985,Petz1986a}.
If the states commute, then a measurement in the common eigenbasis
of the states achieves equality.

It thus remains to prove that the inequality is strict if the states
do not commute. For $\alpha\in\left[1/2,1\right)$, it follows as
a direct consequence of the equality conditions~\cite{Hiai1994Equality}
for the Araki--Lieb--Thirring inequality that
\begin{equation}
D_{\alpha}(\rho\|\sigma)>\widetilde{D}_{\alpha}(\rho\|\sigma)
\end{equation}
in this case. By applying~\eqref{eq:meas-sand-petz}, this implies
that
\begin{equation}
D_{\alpha}(\rho\|\sigma)>D_{\alpha}^{M}(\rho\|\sigma)
\end{equation}
for all $\alpha\in\left[1/2,1\right)$ and positive definite states
$\rho$ and $\sigma$. We can then write this as
\begin{align}
\frac{1}{\alpha-1}\ln\Tr\!\left[\rho^{\alpha}\sigma^{1-\alpha}\right] & >\sup_{\mathcal{X},M_{\mathcal{X}}}\frac{1}{\alpha-1}\ln\sum_{x\in\mathcal{X}}\left(\Tr\!\left[M_{x}\rho\right]\right)^{\alpha}\left(\Tr\!\left[M_{x}\sigma\right]\right)^{1-\alpha}\\
 & =\frac{1}{\alpha-1}\ln\inf_{\mathcal{X},M_{\mathcal{X}}}\sum_{x\in\mathcal{X}}\left(\Tr\!\left[M_{x}\rho\right]\right)^{\alpha}\left(\Tr\!\left[M_{x}\sigma\right]\right)^{1-\alpha},
\end{align}
which in turn implies that
\begin{equation}
\Tr\!\left[\rho^{\alpha}\sigma^{1-\alpha}\right]<\inf_{\mathcal{X},M_{\mathcal{X}}}\sum_{x\in\mathcal{X}}\left(\Tr\!\left[M_{x}\rho\right]\right)^{\alpha}\left(\Tr\!\left[M_{x}\sigma\right]\right)^{1-\alpha}.
\end{equation}
However, since this holds for all positive definite states, we can
exchange $\rho$ and $\sigma$ to conclude that
\begin{equation}
\Tr\!\left[\sigma^{\alpha}\rho^{1-\alpha}\right]<\inf_{\mathcal{X},M_{\mathcal{X}}}\sum_{x\in\mathcal{X}}\left(\Tr\!\left[M_{x}\sigma\right]\right)^{\alpha}\left(\Tr\!\left[M_{x}\rho\right]\right)^{1-\alpha}.
\end{equation}
We can then work backwards to conclude that
\begin{equation}
D_{1-\alpha}(\rho\|\sigma)>D_{1-\alpha}^{M}(\rho\|\sigma).
\end{equation}
So this establishes the strict inequality for all $\alpha\in\left(0,1/2\right]$,
thus concluding the proof.
\end{proof}

\subsection{Proof of Lemma~\ref{lem:one-shot-FD-err-bnd}}

Consider that
\begin{align}
 & \Tr\!\left[\left(e^{-A}+I\right)^{-1}\sigma_{0}\right]+\Tr\!\left[\left(e^{A}+I\right)^{-1}\sigma_{1}\right]\nonumber \\
 & \leq s^{s}\left(1-s\right)^{1-s}\Tr\!\left[e^{\left(1-s\right)A}\sigma_{0}\right]+s^{s}\left(1-s\right)^{1-s}\Tr\!\left[e^{-sA}\sigma_{1}\right]\\
 & =s^{s}\left(1-s\right)^{1-s}\left(\Tr\!\left[e^{\left(1-s\right)\left(\ln\sigma_{1}-\ln\sigma_{0}\right)}\sigma_{0}\right]+\Tr\!\left[e^{-s\left(\ln\sigma_{1}-\ln\sigma_{0}\right)}\sigma_{1}\right]\right)\\
 & =s^{s}\left(1-s\right)^{1-s}p_{0}^{s}p_{1}^{1-s}\left(\Tr\!\left[e^{\left(1-s\right)\left(\ln\rho_{1}-\ln\rho_{0}\right)}\rho_{0}\right]+\Tr\!\left[e^{-s\left(\ln\rho_{1}-\ln\rho_{0}\right)}\rho_{1}\right]\right)\\
 & =g(s,p_{0})\left(\Tr\!\left[e^{\left(1-s\right)\left(\ln\rho_{1}-\ln\rho_{0}\right)}\rho_{0}\right]+\Tr\!\left[e^{s\left(\ln\rho_{0}-\ln\rho_{1}\right)}\rho_{1}\right]\right)\label{eq:big-step-err-bnd-FD-meas}\\
 & =g(s,p_{0})\left(e^{-C_{s}^{\natural}(\rho_{0}\|\rho_{1})}+e^{-C_{1-s}^{\natural}(\rho_{1}\|\rho_{0})}\right)\\
 & \leq2g(s,p_{0})e^{-\min\left\{ C_{s}^{\natural}(\rho_{0}\|\rho_{1}),C_{1-s}^{\natural}(\rho_{1}\|\rho_{0})\right\} }.
\end{align}
The first inequality follows from Lemma~\ref{lem:FD-scalar-bnd-s},
using $1-s$ for the first term and $s$ for the second term.

\subsection{Proof of Proposition~\ref{prop:FD-asymptotic-QHT}}

\label{app:FD-asymptotic-QHT}
\begin{proof}[Proof of Proposition~\ref{prop:FD-asymptotic-QHT}]
Consider that
\begin{align}
 & \Tr\!\left[\left(e^{-A^{(n)}}+I\right)^{-1}\left(p_{0}\rho_{0}^{\otimes n}\right)\right]+\Tr\!\left[\left(e^{A^{(n)}}+I\right)^{-1}\left(p_{1}\rho_{1}^{\otimes n}\right)\right]\nonumber \\
 & \leq2g(s,p_{0})e^{-\min\left\{ C_{s}^{\natural}(\rho_{0}^{\otimes n}\|\rho_{1}^{\otimes n}),C_{1-s}^{\natural}(\rho_{1}^{\otimes n}\|\rho_{0}^{\otimes n})\right\} }\\
 & =2g(s,p_{0})e^{-\min\left\{ nC_{s}^{\natural}(\rho_{0}\|\rho_{1}),nC_{1-s}^{\natural}(\rho_{1}\|\rho_{0})\right\} }\\
 & =2g(s,p_{0})e^{-n\min\left\{ C_{s}^{\natural}(\rho_{0}\|\rho_{1}),C_{1-s}^{\natural}(\rho_{1}\|\rho_{0})\right\} },
\end{align}
where the first inequality follows from Lemma~\ref{lem:one-shot-FD-err-bnd}
and the first equality from Lemma~\ref{lem:additivity-unopt-chernoff}.
\end{proof}
\begin{lem}
\label{lem:additivity-unopt-chernoff}For all $s\in\left[0,1\right]$,
the quantity $C_{s}^{\natural}(\rho_{0}\|\rho_{1})$ is additive in
the following sense:
\begin{equation}
C_{s}^{\natural}(\tau_{0}\otimes\omega_{0}\|\tau_{1}\otimes\omega_{1})=C_{s}^{\natural}(\tau_{0}\|\tau_{1})+C_{s}^{\natural}(\omega_{0}\|\omega_{1}),
\end{equation}
where $\tau_{0}$, $\omega_{0}$, $\tau_{1}$, and $\omega_{1}$ are
positive definite states.
\end{lem}

\begin{proof}
Consider that
\begin{align}
C_{s}^{\natural}(\tau_{0}\otimes\omega_{0}\|\tau_{1}\otimes\omega_{1}) & =-\ln\Tr\!\left[e^{\left(1-s\right)\left(\ln\left(\tau_{1}\otimes\omega_{1}\right)-\ln\left(\tau_{0}\otimes\omega_{0}\right)\right)}\left(\tau_{0}\otimes\omega_{0}\right)\right]\\
 & =-\ln\Tr\!\left[e^{\left(1-s\right)\left(\left(\ln\tau_{1}-\ln\tau_{0}\right)\otimes I+I\otimes\left(\ln\omega_{1}-\ln\omega_{0}\right)\right)}\left(\tau_{0}\otimes\omega_{0}\right)\right]\\
 & =-\ln\Tr\!\left[\left(e^{\left(1-s\right)\left(\ln\tau_{1}-\ln\tau_{0}\right)}\otimes e^{\left(1-s\right)\left(\ln\omega_{1}-\ln\omega_{0}\right)}\right)\left(\tau_{0}\otimes\omega_{0}\right)\right]\\
 & =-\ln\Tr\!\left[e^{\left(1-s\right)\left(\ln\tau_{1}-\ln\tau_{0}\right)}\tau_{0}\otimes e^{\left(1-s\right)\left(\ln\omega_{1}-\ln\omega_{0}\right)}\omega_{0}\right]\\
 & =-\ln\left(\Tr\!\left[e^{\left(1-s\right)\left(\ln\tau_{1}-\ln\tau_{0}\right)}\tau_{0}\right]\Tr\!\left[e^{\left(1-s\right)\left(\ln\omega_{1}-\ln\omega_{0}\right)}\omega_{0}\right]\right)\\
 & =C_{s}^{\natural}(\tau_{0}\|\tau_{1})+C_{s}^{\natural}(\omega_{0}\|\omega_{1}),
\end{align}
thus concluding the proof.
\end{proof}

\subsection{Implementing Fermi--Dirac thermal measurement for i.i.d.~states}

\label{app:Implementing-Fermi=002013Dirac-thermal-iid}

Here we elaborate on a method for realizing the Fermi--Dirac thermal
measurement in~\eqref{eq:n-fold-FD-meas} and~\eqref{eq:additive-form-A-n}
by means of a product measurement followed by classical postprocessing.

Set
\begin{equation}
H\coloneqq\ln\rho_{1}-\ln\rho_{0},
\end{equation}
and let a spectral decomposition of it be as follows:
\begin{equation}
H=\sum_{j}h_{j}|\phi_{j}\rangle\!\langle\phi_{j}|.
\end{equation}
Then
\begin{align}
 & \ln\rho_{1}^{\otimes n}-\ln\rho_{0}^{\otimes n}\nonumber \\
 & =\left(\ln\rho_{1}-\ln\rho_{0}\right)\otimes I^{\otimes n-1}+I\otimes\left(\ln\rho_{1}-\ln\rho_{0}\right)\otimes I^{\otimes n-2}+\cdots+I^{\otimes n-1}\otimes\left(\ln\rho_{1}-\ln\rho_{0}\right)\\
 & =H\otimes I^{\otimes n-1}+I\otimes H\otimes I^{\otimes n-2}+\cdots+I^{\otimes n-1}\otimes H\\
 & \eqqcolon\sum_{i=1}^{n}H^{\left(i\right)}.
\end{align}
Then an eigenbasis for $\ln\rho_{1}^{\otimes n}-\ln\rho_{0}^{\otimes n}$
is $\left\{ |\phi_{j^{n}}\rangle\right\} _{j^{n}}$, where
\begin{equation}
|\phi_{j^{n}}\rangle\coloneqq|\phi_{j_{1}}\rangle\otimes\cdots\otimes|\phi_{j_{n}}\rangle,
\end{equation}
and eigenvector $|\phi_{j^{n}}\rangle$ has eigenvalue $\sum_{i=1}^{n}h_{j_{i}}$.
This follows because
\begin{align}
\sum_{i=1}^{n}H^{\left(i\right)}|\phi_{j^{n}}\rangle & =\sum_{i=1}^{n}H^{\left(i\right)}|\phi_{j_{1}}\rangle\otimes\cdots\otimes|\phi_{j_{n}}\rangle\\
 & =\sum_{i=1}^{n}h_{j_{i}}|\phi_{j_{1}}\rangle\otimes\cdots\otimes|\phi_{j_{n}}\rangle.
\end{align}
Thus,
\begin{equation}
\sum_{i=1}^{n}H^{\left(i\right)}=\sum_{j^{n}}\left(\sum_{i=1}^{n}h_{j_{i}}\right)|\phi_{j_{1}}\rangle\!\langle\phi_{j_{1}}|\otimes\cdots\otimes|\phi_{j_{n}}\rangle\!\langle\phi_{j_{n}}|.
\end{equation}

We conclude that $A^{(n)}$ from~\eqref{eq:additive-form-A-n} can
be written as
\begin{equation}
A^{(n)}=\sum_{j^{n}}\left[\ln\!\left(\frac{p_{1}}{p_{0}}\right)+\sum_{i=1}^{n}h_{j_{i}}\right]|\phi_{j_{1}}\rangle\!\langle\phi_{j_{1}}|\otimes\cdots\otimes|\phi_{j_{n}}\rangle\!\langle\phi_{j_{n}}|,
\end{equation}
and thus
\begin{align}
\left(e^{A^{(n)}}+I^{\otimes n}\right)^{-1} & =\sum_{j^{n}}\left(e^{\ln\left(\frac{p_{1}}{p_{0}}\right)+\sum_{i=1}^{n}h_{j_{i}}}+1\right)^{-1}|\phi_{j_{1}}\rangle\!\langle\phi_{j_{1}}|\otimes\cdots\otimes|\phi_{j_{n}}\rangle\!\langle\phi_{j_{n}}|,\\
\left(e^{-A^{(n)}}+I^{\otimes n}\right)^{-1} & =\sum_{j^{n}}\left(e^{-\left(\ln\left(\frac{p_{1}}{p_{0}}\right)+\sum_{i=1}^{n}h_{j_{i}}\right)}+1\right)^{-1}|\phi_{j_{1}}\rangle\!\langle\phi_{j_{1}}|\otimes\cdots\otimes|\phi_{j_{n}}\rangle\!\langle\phi_{j_{n}}|.
\end{align}

Given this structure, a method for implementing the Fermi--Dirac
thermal measurement in~\eqref{eq:n-fold-FD-meas} consists of
\begin{enumerate}
\item For $i\in\left[n\right]$, measure the $i$th system in the eigenbasis
$\left\{ |\phi_{j}\rangle\right\} _{j}$ of $H$ and record the outcome
as $h_{j_{i}}$.
\item Compute $a(j^{n})\equiv\ln\!\left(\frac{p_{1}}{p_{0}}\right)+\sum_{i=1}^{n}h_{j_{i}}$.
\item Output ``0'' (i.e., the state is $\rho_{0}$) with probability $\left(e^{a(j^{n})}+1\right)^{-1}$,
and output ``1'' (i.e., the state is $\rho_{1}$) with probability
$\left(e^{-a(j^{n})}+1\right)^{-1}$.
\end{enumerate}
The first step above implements a product measurement, and the last
two steps realize classical postprocessing. Under this scheme, for
an arbitrary incoming state $\omega^{(n)}$, the probability of outputting
``0'' is given by
\begin{equation}
\sum_{j^{n}}p(0|j^{n})p(j^{n}),
\end{equation}
where $p(j^{n})$ is the probability to observe the measurement outcome
sequence $j^{n}$ and $p(0|j^{n})$ is the probability of outputting
``0'' after observing $j^{n}$. Given that
\begin{align}
p(j^{n}) & =\Tr\!\left[\left(|\phi_{j_{1}}\rangle\!\langle\phi_{j_{1}}|\otimes\cdots\otimes|\phi_{j_{n}}\rangle\!\langle\phi_{j_{n}}|\right)\omega^{(n)}\right],\\
p(0|j^{n}) & =\left(e^{\ln\left(\frac{p_{1}}{p_{0}}\right)+\sum_{i=1}^{n}h_{j_{i}}}+1\right)^{-1},
\end{align}
we conclude that
\begin{equation}
\sum_{j^{n}}p(0|j^{n})p(j^{n})=\Tr\!\left[\left(e^{A^{(n)}}+I^{\otimes n}\right)^{-1}\omega^{(n)}\right],
\end{equation}
so that the above scheme indeed implements the Fermi--Dirac thermal
measurement. The calculation for the probability of outputting ``1''
is similar.
\end{document}